\documentclass[journal]{IEEEtran}

\IEEEoverridecommandlockouts
\usepackage{cite}
\usepackage{amsmath,amssymb,amsfonts}
\usepackage{graphicx}
\usepackage{textcomp}
\usepackage{algorithm}
\usepackage{algpseudocode}
\usepackage{comment}
\usepackage{mdwtab}
\usepackage{amssymb}
\usepackage{booktabs}
\usepackage{url}
\usepackage{acro}
\usepackage{bbding}
\usepackage{bbm}
\usepackage{blindtext}
\usepackage{enumerate}
\newtheorem{theorem}{Theorem}
\newtheorem{remark}{Remark}
\newtheorem{proof}{Proof}
\def\BibTeX{{\rm B\kern-.05em{\sc i\kern-.025em b}\kern-.08em
    T\kern-.1667em\lower.7ex\hbox{E}\kern-.125emX}}
    
\begin{document}
% \history{Date of publication xxxx 00, 0000, date of current version xxxx 00, 0000.}
% \doi{10.1109/ACCESS.20xx.DOI}

\title{D2D Assisted Multi-Antenna Coded Caching}
% \author{\uppercase{Hamidreza Bakhshzad Mahmoodi}\authorrefmark{1}, \IEEEmembership{Student~Member,~IEEE},
% \uppercase{Jarkko Kaleva}\authorrefmark{2}, \IEEEmembership{Member,~IEEE}, \uppercase{Seyed Pooya Shariatpanahi}\authorrefmark{3}, \uppercase{Antti T\"olli}\authorrefmark{1}, \IEEEmembership{Senior~Member,~IEEE}}
% \address[1]{Centre for Wireless Communications, University of Oulu, FIN-90014 Oulu, Finland (e-mail: Firstname.Lastname@oulu.fi)}
% \address[2]{Solmu Technologies, Oulu, Finland (e-mail: jarkko.kaleva@solmutech.com)}
% \address[3]{School of Electrical and Computer Engineering, College of Engineering, University of Tehran, Iran (e-mail: p.shariatpanahi@ut.ac.ir)}
\author{\IEEEauthorblockN{Hamidreza Bakhshzad Mahmoodi$^\star$, Jarkko Kaleva$^\star$, Seyed Pooya Shariatpanahi$^*$ and Antti T\"olli$^\star$}
\IEEEauthorblockA{
    $\star$ Centre for Wireless Communications, University of Oulu, P.O. Box 4500, 90014, Finland \\
    $*$ School of Electrical and Computer Engineering, College of Engineering, University of Tehran, Iran, \\
    %$\dagger$ Sharif University of Technology, Tehran, Iran \\
    \textrm{firstname.lastname@oulu.fi, p.shariatpanahi@ut.ac.ir} %khalaj@sharif.edu}
    }
    	\thanks{The authors are with Centre for
 			Wireless Communications, University of Oulu P.O. Box 4500, FIN-90014 University
			of Oulu, Finland. Email: \{hamidreza.bakhshzadmahmoodi, antti.tolli, jarkko.kaleva\}@oulu.fi, p.shariatpanahi@ut.ac.ir
			
			Parts of this work has been submitted in the IEEE
			ISWCS, Oulu, Finland, 2019, the IEEE WCNC,
			Seoul, South Korea, 2020.
This work was supported by the Academy of Finland under grants no. 319059 (Coded Collaborative Caching for Wireless Energy Efficiency) and 318927 (6Genesis Flagship).}

}
%\tfootnote{This work was supported by the Academy of Finland under grants no. 319059 (Coded Collaborative Caching for Wireless Energy Efficiency), 318927 (6Genesis Flagship)}

% \markboth
% {H. B. Mahmoodi \headeretal: D2D Assisted Multi-antenna Coded Caching}
% {H. B. Mahmoodi \headeretal: D2D Assisted Multi-antenna Coded Caching}

% \corresp{Corresponding author: Hamidreza Bakhshzad Mahmoodi (e-mail: Hamidreza.bakhshzadmahmoodi@oulu.fi).}

\maketitle

\begin{abstract}
{A device-to-device (D2D) aided multi-antenna coded caching scheme is proposed to improve the average delivery rate and reduce the downlink (DL) beamforming complexity.} Novel beamforming and resource allocation schemes are proposed where local {data exchange} among nearby users is exploited. The transmission is split into two phases: local D2D content exchange and DL transmission. In the D2D phase, subsets of users are selected to share content with the adjacent users directly. {In this regard, a low complexity D2D mode selection algorithm is proposed to find the appropriate set of users for the D2D phase} with comparable performance to the optimal exhaustive search. {During} the DL phase, the base station multicasts the {remaining data requested by all the users}. {We identify scenarios and conditions where D2D transmission can reduce the delivery time.} {Furthermore, we demonstrate how} adding the new D2D phase to the {DL-only scenario} can significantly reduce the beamformer design complexity in the DL phase. The results further highlight that {by partly delivering requested data in the D2D phase, the transmission rate can be boosted due to more efficient use of resources during the subsequent DL phase.} As a result, the overall content delivery performance is greatly enhanced, especially in the finite signal-to-noise (SNR) regime.
\end{abstract}

% \begin{keywords}
% Multi-antenna communications, coded caching, multicast  beamforming, D2D communication. 
% \end{keywords}

% \titlepgskip=-15pt

\section{Introduction}
\label{sec:introduction}
The {daily surge in network traffic has pushed the network providers towards advanced data delivery methods to meet the client's demands. Luckily, a vast share of the network traffic is for multimedia content, such as video streams, which can be proactively cached before hands~\cite{benis5gcache,Paschos2016}.} Specifically, considering the upcoming data-intensive applications such as extended reality (XR), which rely heavily on asynchronous content reuse~\cite{Mahmoodi_immersive_isit2021, Location-dependent-ICC2022,mahmoodi2022asymmetric}, proactive caching at the end-users could relieve network congestion during {busy hours}. The {aim is to leverage} off-peak hours to move some content closer to the end-users, which will be used {to relieve the network's load} during peak hours. Coded caching (CC) has been recently proposed by Maddah-Ali and Niesen in~\cite{MaddahAli-2014}, {which improves the traditional caching schemes~\cite{mehdibennisD2D}  by benefiting from the aggregated memory throughout the network.} As a result, an additional \textit{global caching gain} is achieved on top of traditional local caching gain in the CC scheme.
%instead of locally caching entire files at the end-user, fragments of all the files in the library are stored in the users' cache memories. In the delivery phase, carefully formed coded messages are multicast to groups of users, which results in \textit{global caching gain}~\cite{MaddahAli-2014}. 

The CC scheme {was initially} proposed for the error-free shared link framework (e.g., ~\cite{MaddahAli-2014, Karamchandi2016}, and~\cite{ Pedarsani2016}), but {it was promptly extended to wireless setup}~\cite{Shariatpanahi2016, pooya-cc-physical-2019-journal, Naderalizadeh2017,Naderalizadeh2017-2}.
In this regard, the setup in~\cite{MaddahAli-2014} is extended in~\cite{Shariatpanahi2016} to a multiple server scenario, where the effects of the network structure on the CC performance {are investigated}. Motivated by the results in~\cite{Shariatpanahi2016}, the authors in~\cite{pooya-cc-physical-2019-journal} {propose} different multi-antenna {CC} transmission schemes for the physical layer. {Furthermore, focusing on high signal-to-noise-ratio (SNR) region analysis}, authors in~\cite{Naderalizadeh2017, Naderalizadeh2017-2} show that {CC} boosts the wireless network performance in terms of Degrees-of-Freedom (DoF). Specifically, the global coded caching gain {offered by the CC} and the spatial multiplexing gain {provided by a multi-antenna transmitter} are additive in a wireless broadcast channel~\cite{Shariatpanahi2017}. { Moreover, the multiantenna CC is also highly beneficial with finite SNR as long as the interference is appropriately compensated~\cite{Ngo2017, Piovano2017,  Tolli-Shariatpanahi-Kaleva-Khalaj-Arxiv18}}. Specifically, while \cite{Ngo2017} and \cite{Piovano2017} use a rate-splitting approach to simultaneously benefit from the global caching and the spatial multiplexing gains, %authors in~\cite{Shariatpanahi2017} use zero-forcing (ZF) {beamformers for a multi-antenna setup}.
{a general beamformer design was proposed in~\cite{Tolli-Shariatpanahi-Kaleva-Khalaj-Arxiv18}} to improve the finite-SNR performance of the coded caching in wireless networks.

Despite {significant} CC benefits in wireless networks, several practical challenges have been identified in the literature~\cite{privateCC, D2DsecureCC, CCreducedblocklength, petros_subpack, MJTWC2021, Sharedcache_emanuel, Sharedcache_decenter,  mohamadsubpack, Alpha_fair_CC, UOS_for_CC, Doublecc_ML}. {For instance, users' privacy requirements to prevent information leakage have been specified in}~\cite{privateCC}. In this regard, authors in~\cite{D2DsecureCC} extend the idea in~\cite{privateCC} to a D2D CC network, using non-perfect secret sharing and one-time pad keying techniques. The number of file division requirements in CC, known as \textit{subpacketization} problem, is considered in~\cite{CCreducedblocklength,petros_subpack, MJTWC2021, Sharedcache_emanuel,mohamadsubpack, Sharedcache_decenter}. {While authors in~\cite{CCreducedblocklength} benefit from the linear block code (LBC) technique to reduce the subpacketization, authors in \cite{petros_subpack, MJTWC2021} benefit from the number of antennas at the transmitter. In a similar work, the effect of the subpacketization on the achievable rate in the low-SNR region was also investigated in~\cite{mohamadsubpack}. Finally, authors in\cite{Sharedcache_emanuel, Sharedcache_decenter} propose the shared cache scheme to efficiently control the subpacketization requirement.} % Authors in~\cite{mohamadsubpack} extend the idea in~\cite{petros_subpack} to a general setting where the number of sub-packets can be defined freely among the set of predefined values.

Another crucial issue in the {CC} wireless setting is \textit{near-far} problem. {Explicitly, in conventional CC delivery schemes, a common message is transmitted to several users. Thus, the rate is always limited by the user with the worst channel condition to ensure every user can decode the common message}~\cite{Tolli-Shariatpanahi-Kaleva-Khalaj-Arxiv18}. In this regard, the authors in~\cite{Alpha_fair_CC} propose a congestion control to avoid serving users in a fading state, {while the authors in~\cite{UOS_for_CC} propose to serve users when they are in favorable conditions via appropriate scheduling.} {Using a different approach, authors in~\cite{Doublecc_ML} benefit from deep reinforcement learning (DRL) to jointly design caching and scheduling schemes, tackling the near-far problem}. {Unlike}~\cite{UOS_for_CC} and~\cite{Alpha_fair_CC}, which apply traditional XORing of data elements, {authors in~\cite{tang2017coded} and~\cite{Mahmoodi_immersive_isit2021} allocate different rates to users with different channel conditions by altering the modulation scheme}. Finally, one critical drawback of the traditional wireless CC delivery is the {high complexity of the optimized} beamforming solutions, which can {render multiantenna CC} infeasible for large networks~\cite{Tolli-Shariatpanahi-Kaleva-Khalaj-Arxiv18}. %{The beamforming's complexity issues are addressed in}~\cite{Tolli-Shariatpanahi-Kaleva-Khalaj-Arxiv18}. 

{The CC cache placement must be carefully designed so that every user has some data needed by others~\cite{MaddahAli-2014}}. {On the other hand, in many particular use cases, such as immersive viewing applications, groups of users are being served together~\cite{Mahmoodi_immersive_isit2021, Location-dependent-ICC2022,mahmoodi2022asymmetric}. In such scenarios, users can intermingle at close distances, making the device-to-device (D2D) communication quite alluring.}  {Motivated by this}, a two-phase transmission scheme is proposed in this paper, where the downlink (DL) multicasting of the file fragments in~\cite{Tolli-Shariatpanahi-Kaleva-Khalaj-Arxiv18} is complemented by the direct D2D exchange of local cache contents. %{We show that extending~\cite{Tolli-Shariatpanahi-Kaleva-Khalaj-Arxiv18} to D2D-assisted scenarios improves the average delivery rate and the beamformers' complexity.}

 %%%%%%%%%%%%%%%%%%%%%%%%%%%%%%%%%%
 \subsection{Related works}
 %%%%%%%%%%%%%%%%%%%%%%%%%%%%%%%%%%%%%%%%
 D2D communication has been studied extensively in wireless scenarios for traditional caching methods, where an entire file is placed at the end-user (see for example~\cite{traditional-d2d-magazine-2014, hierarchicald2dcachingsolution, mehdibennisD2D, MDS-CC-D2D-2019}). The idea is that utilizing D2D communication results in less congested servers, improved network energy efficiency (EE), mitigated near-far problems, etc. {On the other hand, currently, interest is increasing in machine learning (ML) tools such as deep deterministic policy gradient (DDPG)~\cite{DRL2020} as a means for various wireless resource allocation problems (e.g.,~\cite{UAV-ML}). In this regard, authors have considered applying ML in cache-aided data delivery, including D2D communications~\cite{ML-D2D-2021}. Specifically, ML tools can be utilized for cache placement or delivery schedules for cache-aided networks~\cite{ML-D2D-2021}, including coded caching~\cite{Doublecc_ML}}.   
 
 {Interestingly}, D2D ideas have been extended to the CC wireless networks to improve network throughput~\cite{Ji2016, D2D-CC-Optload-memtradeof-caire-2019, D2D-CC-H-cachesize-2020, CC-D2D-pooya-2016,  DPDA2019, Centralized-CC-user-cooperation-2019, D2Daimilar2me}. An infrastructure-less {CC} network is considered in~\cite{Ji2016}, where {the D2D is the only transmission link. Moreover} constructive achievable coding strategies and information-theoretic bounds are also specified in~\cite{Ji2016}. Later, the authors in~\cite{D2D-CC-Optload-memtradeof-caire-2019} characterized the optimal load-memory trade-off under the {uncoded placement and one-shot delivery assumptions}. % considered in}~\cite{Ji2016}. 
 Yet, {the authors in~\cite{D2D-CC-Optload-memtradeof-caire-2019} extend the system model in~\cite{Ji2016} such that robustness against random user inactivity is achieved via} maximum distance separation (MDS) coding. {Users with different available memories are considered in~\cite{D2D-CC-H-cachesize-2020} to extend the scenario in~\cite{Ji2016} to a more practical one.} {On the other hand, the authors in~~\cite{CC-D2D-pooya-2016} show that the frequency reuse gain and CC multicasting gain are additive when users move inside the D2D  network.} The authors in~\cite{DPDA2019} extend~\cite{Ji2016} to a general framework {by utilizing the placement delivery array (PDA), alleviating the subpacketization problem.}
 
Authors in~\cite{Centralized-CC-user-cooperation-2019} extend the network in~\cite{MaddahAli-2014} to the case where users can cooperate via D2D links. {Therein, the} D2D and DL transmissions are assumed to be carried out in parallel without interference. Thus, to minimize the total transmission time, the authors in~\cite{Centralized-CC-user-cooperation-2019} optimally divide the transmission load into D2D and DL parts. {Unlike~\cite{Ji2016, D2D-CC-Optload-memtradeof-caire-2019, D2D-CC-H-cachesize-2020, CC-D2D-pooya-2016,  DPDA2019, Centralized-CC-user-cooperation-2019} with fixed rates assumption,} the authors in~\cite{D2Daimilar2me} consider a multi-user single-input-single-output (SISO) scenario where transmission rates are changed depending on channel conditions. In this regard, D2D transmissions are utilized along with the DL transmission to decrease the transmission time. Specifically, the BS decides whether to transmit a sub-file in DL or D2D fashion based on users' neighboring conditions. {Furthermore,} different amounts of memory {are dedicated} to the D2D and DL sub-files {in~\cite{D2Daimilar2me}} to minimize the delivery time.
 %%%%%%%%%%%%%%%%%%%%%%%%%%%%%%%%%%%%%%%%%%%%%%%%%%
\subsection{Main contribution} 
{Although CC-based D2D communication has been thoroughly investigated in the literature for single-antenna transceivers~\cite{Ji2016, D2D-CC-Optload-memtradeof-caire-2019, D2D-CC-H-cachesize-2020, CC-D2D-pooya-2016,  DPDA2019, Centralized-CC-user-cooperation-2019, D2Daimilar2me}, a D2D framework supporting CC-based multi-antenna communication is still missing considering the current multi-antenna transmission support in the 5G standards~\cite{5GNR}. In this regard, this paper extends the framework studied in the previous studies~\cite{Ji2016, D2D-CC-Optload-memtradeof-caire-2019, D2D-CC-H-cachesize-2020, CC-D2D-pooya-2016,  DPDA2019, Centralized-CC-user-cooperation-2019, D2Daimilar2me} to a multi-antenna transmitter scenario, where we consider a single-cell multi-user multiple-input single-output (MISO) network. The proposed delivery scheme in this work comprises D2D links as a complementary transmission phase to the DL transmission.} The cache placement is based on~\cite{MaddahAli-2014}; {thus, every user has some data that can be shared with other users.} The goal is {to minimize the total transmission time by delivering a part of the requested data via D2D transmissions while the rest is delivered by the base station (BS)}. To find the optimal {combination of D2D and DL transmissions, we should be able to compare the available} D2D and the DL rates. However, {the optimal D2D/DL mode selection is a combinatorial problem{, requiring} an exhaustive search for D2D opportunities over a group of users. Hence, due to the NP-hard nature of the problem, the exhaustive search quickly becomes computationally intractable.} On the other hand, to know the exact achievable DL rate, all the beamformers {should} be first designed {for all possible combinations}. {As mentioned before, this is challenging} due to the high computational complexity of multicast beamforming. Consequently, we propose an approximation for the DL achievable rate without computing the beamformers to {resolve the beamforming} complexity issue. Next, based on the approximated DL rate, we provide a low complexity mode selection algorithm, which allows efficient determination of D2D opportunities even for a large number of users. The computational complexity of the proposed algorithm is significantly reduced {compared} to the exhaustive search {with similar} performance. 

{In infrastructure-less networks, as in~\cite{Ji2016}, the geographical separation of users severely affects the total transmission time. However, using the D2D transmission phase as a complementary phase for the DL transmission results in steadier behavior against geographical separation than in the DL-only case in~~\cite{Tolli-Shariatpanahi-Kaleva-Khalaj-Arxiv18} and the D2D-only case in~\cite{Ji2016}. In addition, it significantly shortens the delivery time}. On the other hand, {by} allowing direct D2D exchange of file fragments, the interference management {among a reduced number of} downlink multicast streams becomes more efficient {than in}~\cite{Tolli-Shariatpanahi-Kaleva-Khalaj-Arxiv18}. {Finally}, we thoroughly investigate the complexity reduction of the DL beamformer due to the proposed complementary D2D phase. This paper extends our earlier {conference papers}~\cite{ISWCSmyself,Mahmoodi-etal-Arxiv19} to a general {D2D-aided multiantenna CC}  framework without any group size restrictions and also provides {a more detailed analysis of beamforming complexity}. {Note that this work focuses on instantaneous channel knowledge for D2D/DL mode selection and does not consider any ML-based technique. The reason is that due to numerous practical challenges, satisfactory utilization of ML techniques for cache-aided networks is still missing. A thorough survey on ML utilization challenges for cache-aided networks (e.g., data integrity, dynamic environment, utilization of big data) can be found in~\cite{ML-cacheaided-2020}}.

%%%%%%%%%%%%%%%%%%%%%%%%%%%%%%%%%%%%%%%%%%%%%%%%%%%%
\subsection{Organization and notations}
%%%%%%%%%%%%%%%%%%%%%%%%%%%%%%%%%%%%%%%%%%%%%%%%%%%%
The rest of the paper is organized as follows. In Section \ref{sec:sysmodel} the system model {is introduced}. In Section \ref{sec:examples} we {illustrate} the main {elements} of the proposed {scheme} via two examples. Then, Section~\ref{sec:general} {demonstrates the proposed scheme for the general case}. Section~\ref{sec:complexity_analyze} {analyzes} the complexity reduction of the DL beamforming {due to the complementary D2D phase}. Section \ref{sec:simres} presents numerical results, and Section \ref{sec:conclusions} concludes the paper.
%%%%%%%%%%%%%%%%%%%%%%%%%%%%%%%%%%%%%%%%%%%%%%%%%%%%
\subsubsection*{Notations}
%%%%%%%%%%%%%%%%%%%%%%%%%%%%%%%%%%%%%%%%%%%%%%%%%%%%
Matrices and vectors are presented {in} boldface
upper and lower case letters, respectively. The Hermitian of the matrix $\mathbf{A}$ is denoted as $\mathbf{A}^{\mathrm{H}}$. Cardinality of a discrete set $\mathcal{A}$ is given by $|\mathcal{A}|$. Let $\mathbb{C}$ denote the set of complex numbers and $\|.\|$ be the norm of a complex vector. Also, $[m]$ denotes the set of
integer numbers $\{1, . . . ,m\}$, and $\oplus$ represents addition in the corresponding finite field. {Finally, Table~\ref{table:main notation} collects the main notations used in this paper.}%We denote $\binom{a}{b}=\frac{a!}{(a-b)!b!}$ as the number of combinations of $b$ objects from a set with $a$ objects. In this paper we assume $\binom{a}{b}=0$, for $b>a$.
%%%%%%%%%%%%%%%%%%%%%%%%%%%%%%%%%%%%%%%%%%%%% notation table
\begin{table}[t]
\begin{tabular}{|l|l|l|l|}
 \hline
 \multicolumn{4}{|c|}{List of Notations} \\
 \hline
$\mathcal{W}$ & Library & $Z_k$ & Cache content  \\ \hline
 $\tau$ &  Global caching gain  & $W_{n,\mathcal{V}}$ & Subfile $n$   \\ \hline
 $d_k$ & Request index & $\mathcal{D}$(t) & D2D set \\ \hline
 $\mathcal{R}^{\mathcal{D}}(i)$ & D2D receiver set  & $X^\mathrm{D2D}_i$ & D2D transmit signal \\ \hline
 $P_d$& User transmit power & $N_0$ & Noise variance \\ \hline
$h_{ik}$ & D2D channel & $y_k$ & Received signal \\ \hline
$z_k$ & Receiver noise & $\tilde{X}_{ \mathcal{T}}^{\mathcal{S}}$ & DL transmit message \\ \hline
$\mathbf{w}_{ \mathcal{T}}^{\mathcal{S}}$ & DL beamformer & $\mathbf{h}_k$ & DL channel \\ \hline
$T_\text{D2D}$ & D2D delivery time & $T_\text{DL}$ & DL delivery time \\ \hline
$R_{\text{U}}$ & Symmetric rate & $P_T$& BS transmit power\\ \hline
$\mathbf{x}_{DL}$ & DL transmit signal & $I_{\text{D2D}}(\mathcal{D})$ & D2D user set indicator \\ \hline
$C(.)$ & DL file size & $\overline{\Omega^{\mathcal{S}}}$, $\overline{\Omega^{\mathcal{S},\mathcal{V}}}$ & Set of D2D user sets \\ \hline
 ${\Omega^{\mathcal{S}}}$ & Set of DL messages &  ${\Omega_{k}^{\mathcal{S}}}$ & Needed message set \\ \hline
$\mathcal{I}_k$ & Interference set & $N^{i}_{\mathcal{S}}$, $N^{i}_{\mathcal{V}}$ & DL message count \\ \hline
\end{tabular} \vspace{2mm}
\caption{The list of main notations.}
\label{table:main notation}
\end{table}
%%%%%%%%%%%%%%%%%%%%%%%%%%%%%%%%%%%%%%%%%%%%%

\section{System Model}
\label{sec:sysmodel}

We consider a system consisting of a single $L$ antenna {BS} and $K$ single antenna users. The {BS} has a library of $N$ files, namely $\mathcal{W}=\{W_1,\dots,W_N\}$, where each file has the size of $F$ bits.  {The} normalized cache size (memory) {of each user} is $M$ files. Each user $k$ stores a {fragment of each file}, denoted by $Z_k(W_1,\dots, W_N)$ during the {\it cache placement} phase (cache placement is identical to~\cite{MaddahAli-2014}). {Thus}, for the case $\tau = \frac{KM}{N} \in \mathbb{N}$, each file is divided into $\binom{K}{\tau}$ non-overlapping subfiles, i.e.,
%\begin{equation} \nonumber
$W_n = \{ W_{n,\mathcal{V}} : \mathcal{V} \subset [K], |\mathcal{V}| = \tau\}, \ \forall n \in [N]$.
%\end{equation}
Each user $k$ stores all the subfiles $\mathcal{V} \ni k$, i.e., $Z_k = \{W_{n,\mathcal{V}} \ | \ k \in \mathcal{V},   \mathcal{V} \subset [K],  |\mathcal{V}| = \tau,   \forall n \in [N] \}$. During the {\it content delivery phase}, user $k \in [K]$ makes a request for the file $W_{d_k}, d_k \in [N]$.  
\begin{figure}
   \centering
   \includegraphics[width=1\columnwidth,keepaspectratio]{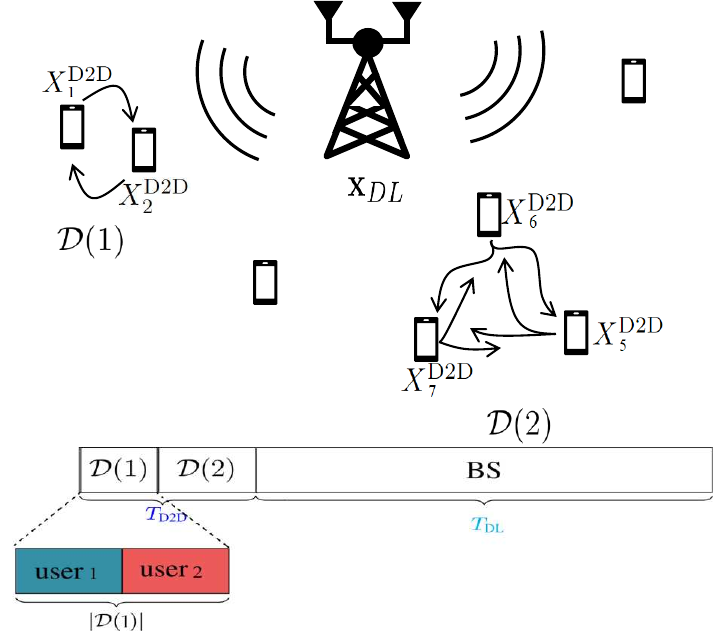}
   \caption{Time division {between} D2D {and DL} transmission slots {for a network with $K=7$ users and two D2D groups $\mathcal{D}(1) = \{1,2\}$ and $\mathcal{D}(2) = \{5,6,7\}$.} The total transmission time to serve all the users is $T_{\text{D2D}}+T_{\text{DL}}$.} \label{fig:timeDiv}
\end{figure}

{We envision a scenario where a group of users is being served simultaneously. These users are free to move inside the coverage area, and at specific times, they request some content from the server. The immersive viewing application illustrated in~\cite{Mahmoodi_immersive_isit2021}, and~\cite{Location-dependent-ICC2022} can be considered a specific use case. Thus, it is very likely to have some nearby users at each time instance who can share some content among themselves via D2D transmissions}. Accordingly, the delivery phase starts with {the D2D} sub-phase, divided into a number of D2D time slots. In each time slot $t$, a group of nearby users, denoted by set $\mathcal{D}(t)$, are instructed by the BS to locally exchange data. Furthermore, each D2D time slot is divided into  $ |\mathcal{D}(t)|$ orthogonal D2D transmissions (see Fig. \ref{fig:timeDiv}). In each D2D transmission, a user $i \in \mathcal{D}(t)$ transmits a coded message denoted by $X^\mathrm{D2D}_i$ to an intended set of receivers $\mathcal{R}^{\mathcal{D}}(i) \subseteq \mathcal{D}(t)$. Thus, the message $X^\mathrm{D2D}_i$ can be transmitted at rate\footnote{In this paper, we assume all D2D user groups $\mathcal{D}(t)$ are served {using time-division multiplexing}. Further improvement can be achieved by allowing parallel mutually interfering transmissions within multiple groups.}
\begin{equation}\label{eq:D2D_multicast_rate}
R_i^{\mathcal{D}} = \min_{k \in \mathcal{R}^{\mathcal{D}}(i)}\log\left( 1+\frac{P_d | h_{ik} |^2}{N_0}\right),
\end{equation} 
where $P_d$ is the user's available transmit power, and $h_{ik} \in \mathbb{C}$ %\sim \mathcal{CN}(0, 1)$ 
is the channel {coefficient} between user $i$ and user $k$. Note that in each D2D transmission, a single user in $\mathcal{D}$ multicasts a message to the rest of the group members. {Hence, the weakest receiver limits the rate.} Furthermore, following~\cite{Ji2016}, {each subfile is transmitted in a D2D group $|\mathcal{D}|-1$ times.} Therefore, subfiles are further divided into $|\mathcal{D}|-1$ smaller parts {for D2D transmissions to ensure fresh content is transmitted}.  

The D2D subphase is followed by the downlink phase, where the {BS} multicasts coded messages containing all the remaining file fragments {so that all users can} decode their requested content. 
The received downlink signal at user terminal $k \in [K]$ is given by \begin{equation}
\label{eq:recv_signal}
    y_k = \mathbf{h}_{k}^{\text{H}} \sum_{ \mathcal{T} \subseteq \mathcal{S}} \mathbf{w}_{ \mathcal{T}}^{\mathcal{S}} \tilde{X}_{ \mathcal{T}}^{\mathcal{S}} + z_k
    \text{,}
\end{equation}
where $\tilde{X}_{ \mathcal{T}}^{\mathcal{S}}$ is the multicast message chosen from a unit power complex Gaussian codebook dedicated to all the users in subset $ \mathcal{T}$ of set $\mathcal{S} \subseteq [K] $ (provided that subset $\mathcal{T}$ has not been considered in {D2D} transmissions). The channel vector between the {BS} and user $k$ is denoted by $\mathbf{h}_k \in \mathbb{C}^L$, the receiver noise is given by $z_k \sim \mathbb{CN}(0, N_0)$, {and $\mathbf{w}_{\mathcal{T}}^{\mathcal{S}}$ is the beamforming vector dedicated to $\tilde{X}_{ \mathcal{T}}^{\mathcal{S}}$}. The channel state information at the transmitter (CSIT) of all $K$ users is assumed to be perfectly known. Finally, the total achievable rate (per user) over the above-described two phases is given by 
\begin{equation}\label{eq:total_rate}
 R_{\text{U}}   =\frac{F}{T_\text{D2D} + T_\text{DL}}
    \text{,}
\end{equation}
where $T_\text{D2D}$ and $T_\text{DL}$ denote the transmission time for the D2D and DL sub-phases, respectively.

\section{D2D Aided Beamforming Explained: Example}
\label{sec:examples}
\begin{figure}
    \centering
    \includegraphics[width=1\columnwidth,keepaspectratio]{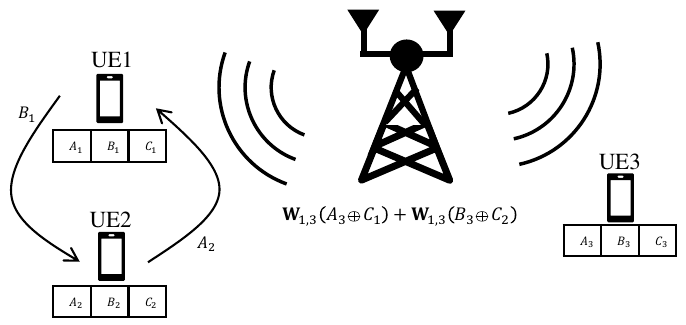}
    \caption{Example 1: D2D enabled downlink beamforming system model, for $K=3, \ L=2, \ \text{and} \ \tau= 1$.}
     \label{fig:sysmodel}
\end{figure}

In this section, we {illustrate} the main {idea} of the proposed {scheme} via two examples. In the first example, we have a network of $3$ users, and in the second example, the number of users is increased to $4$.

\subsection{Example 1: $K=3$, $N=3$, $M=1$, and $L=2$} \label{sec:Example1}
In this example, {the network is comprised of} $K=3$ users and a library $\mathcal{W}=\{A, B, C\}$ of $N=3$ files, where each user has the cache size $M=1$ file. The base station is equipped with $L=2$ transmit antennas. To begin with, the cache content $Z_k$ at each user $k = 1,\ldots,K$ is 
%\begin{equation} \nonumber
    $Z_1 = \{A_1, B_1, C_1\}, \ Z_2 = \{A_2, B_2, C_2\}, \ Z_3 = \{A_3, B_3, C_3\}$.
%\end{equation}
{Note that} each file is divided into three equal-sized sub-files, {following} the cache placement in~\cite{MaddahAli-2014}. In this example, we assume users $1$ and $2$ {are close to each other (c.f. Fig.~\ref{fig:sysmodel})}.

Without loss of generality, assume users $1$, $2$, and $3$ request files $A$, $B$, and $C$, respectively. Now, {the delivery is carried out} in two phases. In the first phase, i.e., the D2D phase, users $1$ and $2$ are assumed to share their local cache content {through the D2D transmission}. Thus, a single D2D time slot with $\mathcal{D}=\{1,2\}$ {exists}. It is evident that user $2$ would request $B_1$ from user $1$, and user $1$ would request $A_2$ from user $2$. The {D2D} transmission is assumed to be half-duplex, i.e., {two uni-directional D2D transmissions are included in each time slot}. The time required for the {D2D} phase is given by 
%\begin{align} \nonumber
    $T_\text{D2D} = T\left(1 \rightarrow \mathcal{R}^{\mathcal{D}}(1)\right)+ T\left(2 \rightarrow \mathcal{R}^{\mathcal{D}}(2)\right)
    = \frac{F / 3}{R^{\mathcal{D}}_1} + \frac{F / 3}{R^{\mathcal{D}}_2}
    \text{,}$ 
%\end{align}
where $\mathcal{R}^{\mathcal{D}}(1)=\{2\}$, $\mathcal{R}^{\mathcal{D}}(2)=\{1\}$, and
%\begin{align} \nonumber
$R^{\mathcal{D}}_1=\log\left(1 + \frac{P_d \lVert h_{12} \rVert^2}{N_0}\right), \
R^{\mathcal{D}}_2=\log\left(1 + \frac{P_d \lVert h_{21} \rVert^2}{N_0}\right)$.
%\end{align}
Note that $\frac{F}{3}$ fractions of the {requested files are delivered in each transmission}.

In the second phase {(i.e., DL transmission)}, the {BS} multicasts the remaining content via coded messages. In the given example, user $3$ {is} not served in the {D2D} phase and still {needs} $C_1$ and $C_2$. However, users $1$ and $2$ only require $A_3$ and $B_3$, respectively. These contents are XOR coded over two messages for user pairs $(1,3)$ and $(2,3)$. Namely, the messages are $X_{1,3} = A_3 \oplus C_1$ and $X_{2,3} = B_3 \oplus C_2$. {Note that} $X_{1,3}$ simultaneously benefits both users $1$ and $3$. Similarly, $X_{2,3}$ is a coded message intended for users $2$ and $3$. {Accordingly}, the multicast beamformer vectors $\mathbf{w}_{1,3}$ and $\mathbf{w}_{2,3}$ are associated with messages $X_{1,3}$ and $X_{2,3}$, respectively. The downlink signal is then formed as
%\begin{equation} \nonumber
    $\mathbf{x}_{DL}=\tilde{X}_{1,3}\mathbf{w}_{1,3}+\tilde{X}_{2,3}\mathbf{w}_{2,3}$, 
%\end{equation}
where $\tilde{X}$ stands for the modulated $X$, chosen from a unit power complex Gaussian codebook \cite{Tolli-Shariatpanahi-Kaleva-Khalaj-Arxiv18}.
Here, user $3$ is assumed to use a successive-interference-cancellation (SIC) receiver to decode both intended messages (interpreted as a multiple-access-channel (MAC)). In contrast, users $1$ and $2$ get served with a single message, with the other seen as interference.  

Now, assume user $3$ can decode \textit{both} messages $\tilde{X}_{1,3}$ and $\tilde{X}_{2,3}$ with the equal rate\footnote{Symmetric rate is imposed to minimize the time needed to receive both messages  $\tilde{X}_{1,3}$ and $\tilde{X}_{2,3}$.} 
%\begin{equation} \nonumber		
$R^3_{MAC}=\min \left(\frac{1}{2} R^3_{Sum}, R^3_1, R^3_2\right)$,  % R^1_{\{1,2\}},R^1_{\{1,3\}})
%\end{equation} 
where the rate region corresponding to $\tilde{X}_{1,3}$ and $\tilde{X}_{2,3}$ is limited by
%\begin{align}
    $R^3_1 =\log\left(1+\frac{|\mathbf{h}_3^H \mathbf{w}_{1,3}|^2} {N_0}\right),
    R^3_2=\log\left(1+\frac{|\mathbf{h}_3^H \mathbf{w}_{2,3}|^2} {N_0}\right)$,
%\end{align}
and
%\begin{equation}
     $R^3_{Sum}=\log\left(1+\frac{|\mathbf{h}_3^H \mathbf{w}_{1,3}|^2+|\mathbf{h}_3^H \mathbf{w}_{2,3}|^2} {N_0}\right)$.
%\end{equation}
Accordingly, the corresponding downlink beamforming problem can be expressed as
%\begin{equation} \nonumber
$\max_{\mathbf{w}_{2,3}, \mathbf{w}_{1,3}}  \min (R^3_\text{MAC}, R^1_1, R^2_1)$,
%\end{equation}
where the rates of users $1$ and $2$ are given as
\begin{align} \nonumber 
    R^1_1&=\log\left(1+\frac{|\mathbf{h}_1^H \mathbf{w}_{1,3}|^2} {|\mathbf{h}_1^H \mathbf{w}_{2,3}|^2+N_0}\right) , \\
    R^2_1&=\log\left(1+\frac{|\mathbf{h}_2^H \mathbf{w}_{2,3}|^2} {|\mathbf{h}_2^H \mathbf{w}_{1,3}|^2+N_0}\right).\nonumber
\end{align}
Here, due to D2D transmissions, the beamforming problem {differs from the one proposed in} ~\cite{Tolli-Shariatpanahi-Kaleva-Khalaj-Arxiv18}. {The partial file exchange in the D2D phase allows focusing the available DL power on fewer multicast streams. In addition, it alleviates the DL phase's inter-stream interference conditions, making the DL multicast beamforming more efficient} and {easier} to design (see Section~\ref{sec:complexity_analyze}). 
On the other hand, the D2D transmission requires an orthogonal allocation in {the} time domain. This introduces an inherent trade-off between the {number} of resources allocated to the D2D and DL phases.

Finally, the corresponding symmetric rate maximization is given as
\begin{equation} \nonumber
\begin{array}{rl}
	\label{prob:dlprob_ex1}
	
	\underset{\substack{\mathbf{w}_{i,j}, \gamma^k_l, r}}{\max} & 
	
        r
	\\
	\mathrm{s.\ t.} 
    &   r \leq \frac{1}{2} \log(1 + \gamma^3_1 + \gamma^3_2), \ 
       r \leq \log(1 + \gamma^3_1), \\& r \leq \log(1 + \gamma^3_2), \
       r \leq \log(1 + \gamma^1_1), \\& r \leq \log(1 + \gamma^2_1), \ \gamma^1_1 \leq \frac{|\mathbf{h}_1^\mathrm{H}\mathbf{w}_{1,3}|^2}{|\mathbf{h}_1^\mathrm{H}\mathbf{w}_{2,3}|^2 + N_0}, \\
          & \gamma^2_1 \leq \frac{|\mathbf{h}_2^\mathrm{H}\mathbf{w}_{2,3}|^2}{|\mathbf{h}_2^\mathrm{H}\mathbf{w}_{1,3}|^2 + N_0} 
    , \
      \gamma^3_1 \leq \frac{|\mathbf{h}_3^\mathrm{H}\mathbf{w}_{1,3}|^2}{N_0},  \\
         & \gamma^3_2 \leq \frac{|\mathbf{h}_3^\mathrm{H}\mathbf{w}_{2,3}|^2}{N_0}, \
        \|\mathbf{w}_{1,3}\|^2 + \|\mathbf{w}_{2,3}\|^2 \leq P_T
		\text{.}
\end{array}
\end{equation}
Where $P_T$ is the total available power at the BS. The rate constraints can be written as a convex second-order cone problem as shown in~\cite{Tolli-Shariatpanahi-Kaleva-Khalaj-Arxiv18}. However, the signal-to-interference-plus-noise
ratio (SINR) constraints are non-convex and require an iterative solution. A successive convex approximation (SCA) solution for the {SINR} constraints can be found, e.g., in~\cite{Tolli-Shariatpanahi-Kaleva-Khalaj-Arxiv18}. Please notice that, here, due to D2D transmission in the first phase, we have only two beamformer vectors ($\mathbf{w}_{1,3}$ and $\mathbf{w}_{2,3}$), which means we can dedicate more power to our intended signals ($X_{1,3}$ and $X_{2,3}$)  compared to~\cite{Tolli-Shariatpanahi-Kaleva-Khalaj-Arxiv18}. The time required for the {DL} phase is given by 
%\begin{equation}
    $T_\text{DL}= \frac{F/3}{r} = \frac{F/3}{\max_{\mathbf{w}_{2,3}, \mathbf{w}_{1,3}}  \min (R^3_\text{MAC}, R^1_1, R^2_1)}
    \text{.}$
%\end{equation}
{Here, similar to the first phase,} all users are served with coded messages of size $\frac{F}{3}$ bits. Finally, the achievable rate over the D2D and DL phases is given in~\eqref{eq:total_rate}.

%%%%%%%%%%%%%%%%%%%%%%%%%%%%%%%%%%%%%%%%%%%%%%%%%%
\subsection{Example 2: $K=4$, $N=4$, $M=2$, and $L=2$}\label{sec:Example2}
%%%%%%%%%%%%%%%%%%%%%%%%%%%%%%%%%%%%%%%%%%%%%%%%%%
\begin{figure}
    \centering
    \includegraphics[width=1\columnwidth,keepaspectratio]{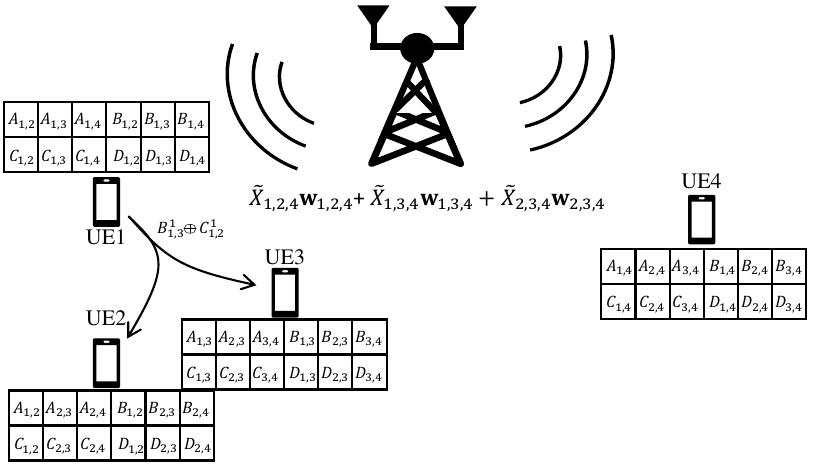}
    \caption{Example 2: D2D enabled downlink beamforming system model, where $K=4, \ L=2, \ \text{and} \ \tau = 2$.}
    \label{fig:sysmodel_complex}
\end{figure}
In this example,  we consider a scenario where $K=4$ users and a library $\mathcal{W}=\{A,B,C,D\}$ of $N=4$ files, {and} each user has a cache for storing $M=2$ files. Also, the base station is equipped with $L=2$ transmit antennas. Following the cache placement in~\cite{MaddahAli-2014}, each file is split into $\binom{K}{\tau}=\binom{4}{2}=6$ subfiles as follows
%%%%%%%%%%%%%%%%%%%%%%%%%%%%%%%%%%%%%%%%%%%%%%%%%%%%%%%%%%%%%%%%%%%% two column
\begin{align} \nonumber
A&=\{A_{1,2}, A_{1,3}, A_{1,4}, A_{2,3}, A_{2,4}, A_{3,4}\}, \\ \nonumber
B&=\{B_{1,2}, B_{1,3}, B_{1,4}, B_{2,3}, B_{2,4}, B_{3,4}\}, \\ \nonumber
C&=\{C_{1,2}, C_{1,3}, C_{1,4}, C_{2,3}, C_{2,4}, C_{3,4}\}, \\ \nonumber
D&=\{D_{1,2}, D_{1,3}, D_{1,4}, D_{2,3}, D_{2,4}, D_{3,4}\}.
\end{align}
%%%%%%%%%%%%%%%%%%%%%%%%%%%%%%%%%%%%%%%%%%%%%%%%%%%%%%%%%%%%%%%%%%%% one column
% \begin{align} \nonumber
% A&=\{A_{1,2}, A_{1,3}, A_{1,4}, A_{2,3}, A_{2,4}, A_{3,4}\}, \quad \nonumber
% B=\{B_{1,2}, B_{1,3}, B_{1,4}, B_{2,3}, B_{2,4}, B_{3,4}\}, \\ \nonumber
% C&=\{C_{1,2}, C_{1,3}, C_{1,4}, C_{2,3}, C_{2,4}, C_{3,4}\}, \quad \nonumber
% D=\{D_{1,2}, D_{1,3}, D_{1,4}, D_{2,3}, D_{2,4}, D_{3,4}\}.
% \end{align}
Each file $W_{\mathcal{T}}$ is cached at user $k$ if $k \in \mathcal{T}$. Let us assume that users $1-4$ request files $A-D$, respectively.

In this example, we {assume} users $1$, $2$, and $3$ are close to each other, while user $4$ is far from them as illustrated in Fig.~\ref{fig:sysmodel_complex}. 
Then, {during this phase, the first three users (collected in $\mathcal{D}=\{1,2,3\}$) locally exchange content in three orthogonal D2D transmissions}.
Following~\cite{Ji2016}, each subfile is further divided into $|\mathcal{D}(t)|-1 = 2$ fragments, discriminated by their superscript indices. Then, in the first D2D transmission of length $T\left(1\rightarrow \mathcal{R}^{\mathcal{D}}(1)\right)$ seconds, user $1$ multicasts $X^{D2D}_1= B^{1}_{1,3} \oplus C^{1}_{1,2}$ to $\mathcal{R}^{\mathcal{D}}(1)=\{2,3\}$. In the second D2D transmission, user $2$ transmits $X^{D2D}_2= A^{1}_{2,3} \oplus C^{2}_{1,2}$ to $\mathcal{R}^{\mathcal{D}}(2)=\{1,3\}$, which will take $T\left(2\rightarrow \mathcal{R}^{\mathcal{D}}(2)\right)$ seconds. Finally, in the third D2D transmission of length $T\left(3\rightarrow \mathcal{R}^{\mathcal{D}}(3)\right)$ seconds, user $3$ transmits $X^{D2D}_3= A^{2}_{2,3} \oplus B^{2}_{1,3}$ to $\mathcal{R}^{\mathcal{D}}(3)=\{1,2\}$. 
These transmissions require the total time of
%\begin{equation}
%\begin{array}{rl}
    $T_{D2D} = T\left(1\rightarrow \mathcal{R}^{\mathcal{D}}(1)\right)+T\left(2\rightarrow \mathcal{R}^{\mathcal{D}}(2)\right)+
    T\left(3\rightarrow \mathcal{R}^{\mathcal{D}}(3)\right)$,
%\end{array}
%\end{equation}
in which
\begin{math}
T\left(i\rightarrow \mathcal{R}^{\mathcal{D}}(i)\right) = \frac{F/12}{R^{\mathcal{D}}_i}, \quad i= {1,2,3}
\end{math}
and $R^{\mathcal{D}}_i, i={1,2,3}$ is given in~\eqref{eq:D2D_multicast_rate}.

In the DL phase, the BS transmits a message comprised of the remaining subfiles
%\begin{equation}
$\mathbf{x}_{DL}=\tilde{X}_{1,2,4} \mathbf{w}_{1,2,4}+\tilde{X}_{1,3,4} \mathbf{w}_{1,3,4}+\tilde{X}_{2,3,4} \mathbf{w}_{2,3,4},$
%\end{equation}
%\begin{equation}
%\mathbf{x}_{DL}=\tilde{X}_{4} \mathbf{w}_{4},
%\end{equation}
where ${X}_{1,2,4}=A_{2,4} \oplus B_{1,4} \oplus D_{1,2}$, ${X}_{1,3,4}=A_{3,4} \oplus C_{1,4} \oplus D_{1,3}$, and ${X}_{2,3,4}=B_{3,4} \oplus C_{2,4} \oplus D_{2,3}$\footnote{For convenience, the superscript $\mathcal{S} = \{ 1,2,3,4 \}$ in $\mathbf{w}_{ \mathcal{T}}^{\mathcal{S}} \tilde{X}_{ \mathcal{T}}^{\mathcal{S}}$ has been omitted in this example.}, and $\tilde{X}_{\mathcal{T}}$ is the modulated version of $X_{\mathcal{T}}$. At the end of this phase, user $1$ is interested in decoding $\{X_{1,2,4}$, $X_{1,3,4}\}$, user $2$ is interested in decoding $\{X_{1,2,4}$, $X_{2,3,4}\}$, user $3$ is interested in decoding $\{X_{1,3,4}$, $X_{2,3,4}\}$, and, user $4$ is interested in decoding all the three terms $\{X_{1,2,4}, X_{1,3,4}$, $X_{2,3,4}\}$. Thus, from the perspective of users $1$, $2$, and $3$, {there exists a} {MAC} channel with two useful terms and one interference term. However, from the perspective of the user $4$, there exists a {MAC} channel with three useful terms. In this regard, for users $1$, $2$, and $3$ the {MAC} rate region {is given as}
%\begin{equation}		
$R^k_\text{MAC}=\min (\frac{1}{2}R^k_\text{sum},  R^k_1,  R^k_2), \ k=1,2,3$.
%\end{equation} 
For instance, for user $1$, we have
%\begin{align} \nonumber
    $R^1_1 =\log\left(1+\frac{|\mathbf{h}_1^{\mathrm{H}}\mathbf{w}_{1,2,4}|^2}{|\mathbf{h}_1^{\mathrm{H}}\mathbf{w}_{2,3,4}|^2 + N_0}\right), % \\ \nonumber
    R^1_2 =\log\left(1+\frac{|\mathbf{h}_1^{\mathrm{H}}\mathbf{w}_{1,3,4}|^2}{|\mathbf{h}_1^{\mathrm{H}}\mathbf{w}_{2,3,4}|^2 + N_0}\right)$,
%\end{align}
    and
%\begin{align} \nonumber
    $R^1_\text{sum}= \log \left(1+\frac{|\mathbf{h}_1^{\mathrm{H}}\mathbf{w}_{1,2,4}|^2+|\mathbf{h}_1^{\mathrm{H}}\mathbf{w}_{1,3,4}|^2}{|\mathbf{h}_1^{\mathrm{H}}\mathbf{w}_{2,3,4}|^2 + N_0}\right)$.
%\end{align}
%The expressions for the 2nd and 3rd users can be similarly derived.

To derive the fourth user's rate region, we face a MAC with three messages. Thus, we have $7$ MAC region inequalities, which will result in $R^4_\text{MAC}$. Then, the corresponding rate constraints for the {MAC} channel are listed below
\begin{equation} \nonumber \scriptsize
    \begin{aligned}
         &R^4_1=\log\left(1+\frac{|\mathbf{h}_4^{\mathrm{H}}\mathbf{w}_{1,2,4}|^2}{N_0}\right), \ %\ 
         %&
         R^4_2=\log\left(1+\frac{|\mathbf{h}_4^{\mathrm{H}}\mathbf{w}_{1,3,4}|^2}{N_0}\right), \\
         &
         R^4_3=\log\left(1+\frac{|\mathbf{h}_4^{\mathrm{H}}\mathbf{w}_{2,3,4}|^2}{N_0}\right), \\ &R^4_{1,2}=\log\left(1+\frac{|\mathbf{h}_4^{\mathrm{H}}\mathbf{w}_{1,2,4}|^2 + |\mathbf{h}_4^{\mathrm{H}}\mathbf{w}_{1,3,4}|^2}{N_0}\right), \\
         & 
         R^4_{1,3}=\log\left(1+\frac{|\mathbf{h}_4^{\mathrm{H}}\mathbf{w}_{1,2,4}|^2 + |\mathbf{h}_4^{\mathrm{H}}\mathbf{w}_{2,3,4}|^2}{N_0}\right), \\
         &R^4_{2,3}=\log\left(1+\frac{|\mathbf{h}_4^{\mathrm{H}}\mathbf{w}_{1,3,4}|^2 + |\mathbf{h}_4^{\mathrm{H}}\mathbf{w}_{2,3,4}|^2}{N_0}\right), \\
         & 
         \\ &R^4_{1,2,3}=\log\left(1+\frac{|\mathbf{h}_4^{\mathrm{H}}\mathbf{w}_{1,2,4}|^2 + |\mathbf{h}_4^{\mathrm{H}}\mathbf{w}_{1,3,4}|^2 + |\mathbf{h}_4^{\mathrm{H}}\mathbf{w}_{2,3,4}|^2}{N_0}\right).
    \end{aligned}
\end{equation}
Thus, the {MAC} rate region for user $4$ is expressed as follows
\begin{equation} \nonumber
R^4_\text{MAC}=\min (\frac{1}{3}R^4_\text{1,2,3},\frac{1}{2}R^4_\text{1,2},\frac{1}{2}R^4_\text{1,3}, \frac{1}{2}R^4_\text{2,3},  R^4_1,  R^4_2,R^4_3).
\end{equation}
When all the MAC inequalities for all the users are considered together, {the common multicast rate is driven as follows}
%%%%%%%%%%%%%%%%%%%%%%%%%%%%%%%%%%%%%%%%%%%%%%%%%%%%%%%%%% one column
\begin{equation} \small
\begin{array} {l}
	\label{prob:dlprob_ex2}
		\underset{\substack{\mathbf{w}_{i,j,l}, \gamma^k_m, r }}{\max}  
	
        r
	\\
	 \mathrm{subject  \ to} \\ 
      r \leq \frac{1}{2}\log(1 + \gamma^k_1 + \gamma^k_2),\ k=1,2,3, \\ %sum rate for users 1,2,3
      r \leq \log(1 + \gamma^k_m),\ k=1,2,3, m = 1,2, \\
     r \leq \frac{1}{3}\log(1 + \gamma^4_1 + \gamma^4_2+\gamma^4_3),  \ % sum rate for user 4 
      r \leq \frac{1}{2}\log(1 + \gamma^4_1 + \gamma^4_2) ,\\ % 2 stream sum rates
     r \leq \frac{1}{2}\log(1 + \gamma^4_1 + \gamma^4_3), \ % 2 stream sum rates
     r \leq \frac{1}{2}\log(1 + \gamma^4_2 + \gamma^4_3), \\ % 2 stream sum rates
      r \leq \log(1 + \gamma^4_m),\ m = 1,2,3 \\
     \gamma^1_1 \leq \frac{|\mathbf{h}_1^{\mathrm{H}}\mathbf{w}_{1,2,4}|^2}{|\mathbf{h}_1^{\mathrm{H}}\mathbf{w}_{2,3,4}|^2 + N_0}, 
         \gamma^1_2 \leq \frac{|\mathbf{h}_1^{\mathrm{H}}\mathbf{w}_{1,3,4}|^2}{|\mathbf{h}_1^{\mathrm{H}}\mathbf{w}_{2,3,4}|^2 + N_0} ,
    \\
     \gamma^2_1 \leq \frac{|\mathbf{h}_2^{\mathrm{H}}\mathbf{w}_{1,2,4}|^2}{|\mathbf{h}_2^{\mathrm{H}}\mathbf{w}_{1,3,4}|^2 + N_0}, 
         \gamma^2_2 \leq \frac{|\mathbf{h}_2^{\mathrm{H}}\mathbf{w}_{2,3,4}|^2}{|\mathbf{h}_2^{\mathrm{H}}\mathbf{w}_{1,3,4}|^2 + N_0} ,
    \\
     \gamma^3_1 \leq \frac{|\mathbf{h}_3^{\mathrm{H}}\mathbf{w}_{1,3,4}|^2}{|\mathbf{h}_3^{\mathrm{H}}\mathbf{w}_{1,2,4}|^2 + N_0}, 
         \gamma^3_2 \leq \frac{|\mathbf{h}_3^{\mathrm{H}}\mathbf{w}_{2,3,4}|^2}{|\mathbf{h}_3^{\mathrm{H}}\mathbf{w}_{1,2,4}|^2 + N_0} ,
    \\
     \gamma^4_1 \leq |\mathbf{h}_4^{\mathrm{H}}\mathbf{w}_{1,2,4}|^2/N_0, 
         \gamma^4_2 \leq |\mathbf{h}_4^{\mathrm{H}}\mathbf{w}_{1,3,4}|^2/N_0, \\
    
         \gamma^4_3 \leq |\mathbf{h}_4^{\mathrm{H}}\mathbf{w}_{2,3,4}|^2/N_0, \\
	  
        \|\mathbf{w}_{1,2,4}\|^2 + \|\mathbf{w}_{1,3,4}\|^2 + \|\mathbf{w}_{2,3,4}\|^2 \leq P_T
		\text{.}
\end{array}
\end{equation}
Finally, the delivery time of the DL phase is
%\begin{equation}
    $T_\text{DL} = \frac{F/6}{r}$.
    %\text{.}
%\end{equation}

It should be noted that compared to~\cite{Tolli-Shariatpanahi-Kaleva-Khalaj-Arxiv18}, one term is removed from the DL transmission {herein}, i.e., $\tilde{X}_{1,2,3}\mathbf{w}_{1,2,3}$. We have taken care of this term in the D2D phase, {enhancing the performance of the DL phase} for two reasons. First, since we removed one term from DL transmission, {the transmit power is shared by fewer beamformers. Second, since one term is removed, the number of conditions in the optimization problem is less than}~\cite{Tolli-Shariatpanahi-Kaleva-Khalaj-Arxiv18}. This will reduce the complexity of the optimization problem %\footnote{An alternative approach to further decrease the complexity of the beamformer design is proposed in~\cite{Tolli-Shariatpanahi-Kaleva-Khalaj-Arxiv18}. The remaining three messages can be transmitted separately in three orthogonal time slots, which also results in complexity reduction of the beamformer with the expense of some minor rate loss (refer to~\cite{Tolli-Shariatpanahi-Kaleva-Khalaj-Arxiv18} for more details).} 
as discussed in Sec.~\ref{sec:complexity_analyze}.

%%%%%%%%%%%%%%%%%%%%%%%%%%%%%%%%%%%%%%%%%%%%%%%%%%%%%%%%%%%%%%%%%%%%%
\subsection{Example 2: D2D group sizes less than $\tau+1$} \label{sec:example-2_general_group_size}
%%%%%%%%%%%%%%%%%%%%%%%%%%%%%%%%%%%%%%%%%%%%%%%%%%%%%%%%%%%%%%%%%%%%%%%%%%%%%
So far, based on the {scheme} proposed in~\cite{Ji2016}, we have considered D2D group size $|\mathcal{D}| = \tau+1 = 3$. However, {given the scenario presented in Fig.~\ref{fig:sysmodel_complex}, there are still some useful contents in the cache of the three users, which can be shared among them in D2D groups of size} $|\mathcal{D}| = 2$, i.e., $\mathcal{D}(1)= \{1,2\}, \ \mathcal{D}(2)= \{1,3\}, \ \mathcal{D}(3)= \{2,3\}$. In this regard, user $1$ transmits $B_{1,4} \ \text{and} \ C_{1,4}$ to users $2$ and $3$, user $2$ transmits $A_{2,4} \ \text{and} \ C_{2,4}$ to users $1$ and $3$, and user $3$ transmits $A_{3,4}$ and $B_{3,4}$ to users $1$ and $2$, respectively. Therefore, compared to section~\ref{sec:Example2}, six more time slots are needed in the D2D phase in this case, where the transmission scheme is similar to section~\ref{sec:Example1}. Then, since user $4$ has not received any data in the D2D phase, it still needs to receive all three missing subfiles through DL transmission. On the other hand, users $1$, $2$, and $3$ have received {all the missing} data through D2D transmissions. Thus, the downlink transmission is changed to $\mathbf{x}_{DL}= \tilde{X}_{1,2,4} \mathbf{w}_{1,2,4}+\tilde{X}_{1,3,4} \mathbf{w}_{1,3,4}+\tilde{X}_{2,3,4} \mathbf{w}_{2,3,4}$, where $\tilde{X}_{1,2,4}= D_{1,2}$, $\tilde{X}_{1,3,4}=D_{1,3}$, and $\tilde{X}_{2,3,4}=D_{2,3}$. The rate expression can be formulated similar to~\eqref{prob:dlprob_ex2} with seven MAC conditions comprised of $\gamma^{4}_{1}, \ \gamma^{4}_{2}, \ \text{and} \ \gamma^{4}_{3}$. However, {since the BS only serves one user in this special case, the DL phase can be simplified to unicast transmission.} To this end, the DL message can be expressed as $\hat{\mathbf{x}}_{DL} ={\mathbf{w}}_4\tilde{X}_4 $, where $\tilde{X}_4 = [D_{1,2}, D_{1,3}, D_{2,3}]$ is the concatenated version of the three missing parts with the total size of $\frac{F}{2}$ bits. Then, the corresponding beamformer can be expressed as a maximum ratio transmitter (MRT), i.e., ${\mathbf{w}}^{*}_{4} = \frac{{\mathbf{h}}_{4}P_{T}}{|{\mathbf{h}}_4|}$. Note that, though the number of messages in the $\mathbf{x}_{DL}$ (in this case) remains the same as in Section~\ref{sec:Example2}, the complexity of the beamformer design is greatly reduced. As a result, the rate for DL transmission is further enhanced, which in turn leads to a potential reduction of the total transmission time.

\section{D2D Aided Beamforming: The General Case}
\label{sec:general}
In this section,  we formulate and analyze the proposed scheme in the general setting. We first consider D2D group size $|\mathcal{D}|=\tau +1$, then in section~\ref{sec: D2D-general-groupsize-LTP1}, we extend the results to group sizes smaller than $\tau +1$.
The cache placement phase is identical to the one proposed in \cite{MaddahAli-2014}. In general, similar to~\cite{Tolli-Shariatpanahi-Kaleva-Khalaj-Arxiv18}, $\min(\tau+L, K)$ users are served in each data transmission. However, {in our proposed scheme, the data is delivered over separate D2D and DL phases.}

{Before the delivery process, an exhaustive search among the D2D subsets is required to find the potential D2D groups for the first phase}. However, since there are $\binom{\tau+L}{\tau+1}$ different D2D subsets (of size $\tau+1$) among $\tau+L$ number of users, the exhaustive search would require $2^{\binom{\tau+L}{\tau+1}}$ evaluations of~\eqref{eq:total_rate}. Moreover, all the beamformers must be solved in each of these evaluations, and the total rate must be computed. Then, the case with the lowest delivery time is selected. However, this is practically infeasible due {to the significant} overhead. Thus, to reduce the computational burden, a less complex heuristic mode selection method {is} introduced in Section~\ref{section: general-partB}. To simplify the notation, we consider an indicator function $I_{D2D}(\mathcal{D})$, which specifies whether the corresponding subset has been allocated for D2D transmission. Moreover,
we denote $C(K,\tau,L)=\frac{F}{\binom{K}{\tau}\binom{K-(\tau+1)}{L-1}}$ as the size of the transmitted subfiles~\cite{Tolli-Shariatpanahi-Kaleva-Khalaj-Arxiv18}.
{\subsection{Total delivery time $T_{\mathrm{D2D}}+T_{\mathrm{DL}}$}\label{sec:delivery_time_general} \label{sec: general-parta}
%%%%%%%%%%%%%%%%%%%%%%%%%%%%%%%%%%%%%%%%%%%%%%%%%%%%%%%%%%%%%%%%%%%%
{First}, the D2D delivery time {for a given selection of D2D subsets} is given as  
\begin{align}
    &T_{\mathrm{D2D}}  = \sum_{\mathcal{D} \subseteq \overline{\Omega^{\mathcal{S}}}}\sum_{k \in \mathcal{D}}\frac{\frac{C(K,\tau,L)}{(|\mathcal{D}|-1)}}{{R}^{\mathcal{D}}_{k}}\label{Eq:pd2d},
\end{align}
where}
\begin{math} %\label{eq: usesub}
    \overline{\Omega^{\mathcal{S}}} := \{\mathcal{D} \subseteq \mathcal{S}, |\mathcal{D}|=\tau+1, I_{\mathrm{D2D}}(\mathcal{D})=1\}
\end{math}, and ${R}^{\mathcal{D}}_{k}$ is given in~\eqref{eq:D2D_multicast_rate}.
Since in each D2D subset, each user transmit a $\frac{1}{|\mathcal{D}|-1}$ fraction of a subfile, the corresponding data size in each D2D transmission is $\frac{C(K,\tau,L)}{(|\mathcal{D}|-1)}$ (see section~\ref{sec:Example2}).

{Next}, the DL beamforming is done {using the SCA approach proposed in}~\cite{Tolli-Shariatpanahi-Kaleva-Khalaj-Arxiv18}. Compared to~\cite{Tolli-Shariatpanahi-Kaleva-Khalaj-Arxiv18}, we {do} not consider all the $\tau+1$ subsets. Here, subsets $\mathcal{D}$ for which $I_{\mathrm{D2D}}(\mathcal{D})=0$ is {considered} in the DL phase, {reducing} the interference among parallel streams significantly. The DL phase throughput is given by
%%%%%%%%%%%%%%%%%%%%%%%%%%%%%%%%%%%%%%%%%%%%%%%%%%%%%%%%% one column
%	\begin{align}\label{Eq:MulticastRate} \nonumber
%&R_C\left(\mathcal{S},\{\mathbf{w}_\mathcal{D}^\mathcal{S}, \mathcal{D} \subseteq \mathcal{S}, |\mathcal{D}|=\tau+1, I_{D2D}(\mathcal{D})=0\}\right)= \\
%& \min_{k \in \mathcal{S}} R^k_{MAC}\left(\mathcal{S},\{\mathbf{w}_\mathcal{D}^\mathcal{S}, \mathcal{D} \subseteq \mathcal{S}, %|\mathcal{D}|=t+1, 
%I_{\mathrm{D2D}}(\mathcal{D})=0\}\right)
%	\end{align}
%%%%%%%%%%%%%%%%%%%%%%%%%%%%%%%%%%%%%%%%%%%%%%%%%%%%%%%%%%%

\begin{align} \label{Eq:MulticastRate} 
&R_C\left(\mathcal{S},\{\mathbf{w}_\mathcal{D}^\mathcal{S}, \mathcal{D} \subseteq \mathcal{S}, |\mathcal{D}|=\tau+1, I_{D2D}(\mathcal{D})=0\}\right)= \nonumber\\
& \quad \quad \max_{\{\mathbf{w}_\mathcal{D}^\mathcal{S}\}}\min_{k \in \mathcal{S}} R^k_{MAC}\left(\mathcal{S},\{\mathbf{w}_\mathcal{D}^\mathcal{S}, \mathcal{D} \subseteq \mathcal{S}, %|\mathcal{D}|=t+1, 
I_{\mathrm{D2D}}(\mathcal{D})=0\}\right),
	\end{align}
and
%%%%%%%%%%%%%%%%%%%%%%%%%%%%%%%%%%%%%%%%%%%%%%%%%% one column
\begin{align}\label{Eq:MAC_general}
	&R^k_{MAC}\left(\mathcal{S},\{\mathbf{w}_\mathcal{D}^\mathcal{S}, \mathcal{D} \subseteq \mathcal{S}, %|\mathcal{D}|=t+1, 
	I_{\mathrm{D2D}}(\mathcal{D})=0\}\right) \nonumber \\ 
	%= &\min_{k \in \mathcal{S}}\min \left[\frac{1}{|B|}R_k^\mathcal{S} , \forall {B \subseteq \Omega^k_S}\right],\\
	= &\min_{\mathcal{B} \subseteq \Omega_k^{\mathcal{S}}} \left[\frac{1}{|\mathcal{B}|}\log\left(1+\frac{\sum_{\mathcal{D} \in \mathcal{B}} |\mathbf{h}_k^{\mathrm{H}} \mathbf{w}_\mathcal{D}^\mathcal{S}|^2}{N_0+\sum_{\mathcal{D'} \in \mathcal{I}_{k}} |\mathbf{h}_k^{\mathrm{H}} \mathbf{w}_\mathcal{D'}^\mathcal{S}|^2}\right) \right],
		%	\Omega_\mathcal{S} \backslash \Omega_k^\mathcal{S}
		\end{align}
%%%%%%%%%%%%%%%%%%%%%%%%%%%%%%%%%%%%%%%%%%%%%%%%%%
% \begin{flalign}\label{Eq:MAC_general}
% 	R^k_{MAC}\left(\mathcal{S},\{\mathbf{w}_\mathcal{D}^\mathcal{S}, \mathcal{D} \subseteq \mathcal{S}, %|\mathcal{D}|=t+1, 
% 	I_{\mathrm{D2D}}(\mathcal{D})=0\}\right)  
% 	%= &\min_{k \in \mathcal{S}}\min \left[\frac{1}{|B|}R_k^\mathcal{S} , \forall {B \subseteq \Omega^k_S}\right],\\
% 	= &\min_{\mathcal{B} \subseteq \Omega_k^{\mathcal{S}}} \left[\frac{1}{|\mathcal{B}|}\log\left(1+\frac{\sum_{\mathcal{D} \in \mathcal{B}} |\mathbf{h}_k^{\mathrm{H}} \mathbf{w}_\mathcal{D}^\mathcal{S}|^2}{N_0+\sum\limits_{\mathcal{D'} \in \mathcal{I}_{k}} |\mathbf{h}_k^{\mathrm{H}} \mathbf{w}_\mathcal{D'}^\mathcal{S}|^2}\right) \right],
% 		%	\Omega_\mathcal{S} \backslash \Omega_k^\mathcal{S}
% 		\end{flalign}
		where 
%		\begin{align}
$\mathcal{I}_{k} = \Omega^{\mathcal{S}} \setminus \Omega_k^{\mathcal{S}} := \{\mathcal{D} \subseteq \mathcal{S} : |\mathcal{D}|=\tau+1, \ I_{\mathrm{D2D}}(\mathcal{D})=0 \ | \ k \notin \mathcal{D}\}$ is the set of {interfering messages} at user $k$. Denote $\Omega^{\mathcal{S}} := \{\mathcal{D} \subseteq \mathcal{S}, |\mathcal{D}|=\tau+1, I_{\mathrm{D2D}}(\mathcal{D})=0\}$ as the set of all the user subsets (of size $\tau+1$) that will be served in the DL phase\footnote{In Example 2, $\Omega^{\mathcal{S}}= \{\{1,2,4\}, \{1,3,4\}, \{2,3,4\}\}$ and $|\Omega^{\mathcal{S}}|=3$.}, where the cardinality $|\Omega^{\mathcal{S}}|$ indicates the total number of messages delivered by the BS. Finally, let
$\Omega_k^{\mathcal{S}} := \{\mathcal{D} \subseteq \mathcal{S}, |\mathcal{D}|=\tau+1, I_{\mathrm{D2D}}(\mathcal{D})=0 \ | \ k \in \mathcal{D}\}$ denote the set of all the subsets in which user $k$ exists (i.e., the set of all the messages required by user $k$).
%		\end{align}
%\rev{Define $T_\mathrm{DL}$ here, and refer to (3)} 

After computing the rate for the DL phase, $T_\mathrm{DL}$ is given as $T_\mathrm{DL}=\frac{C(K,\tau,L)}{R_C}$, while the symmetric delivery rate is given in~\eqref{eq:total_rate}. 
Note that solving~\eqref{Eq:MulticastRate} requires considerable computation due to the iterative convex approximation process and many constraints. Moreover, the exhaustive search would require computing~\eqref{Eq:pd2d} and~\eqref{Eq:MulticastRate} for each D2D subset evaluation. Therefore, considering the total number of different D2D mode allocations (i.e., different combinations of subsets) and {the complexity of} computing~\eqref{Eq:MulticastRate} for each of these modes, the exhaustive search becomes impractical for large networks. Therefore, {in the following}, we provide a low-complexity heuristic solution for the proposed mode assessment problem.
%%%%%%%%%%%%%%%%%%%%%%%%%%%%%%%%%%%%%%%%%%%%%%%%%%%%%%%%%%%%%%%%%%%%
\subsection{Heuristic D2D mode selection with low complexity} \label{section: general-partB}
%%%%%%%%%%%%%%%%%%%%%%%%%%%%%%%%%%%%%%%%%%%%%%%%%%%%%%%%%%%%%%%%%%%%
To decrease the computational load of evaluating ${T}_{\mathrm{D2D}}$ and ${T}_{\mathrm{DL}}$ for different D2D mode allocations, we provide a throughput approximation for the D2D mode allocations without relying on the general SCA solution for the DL beamformer design.
On the one hand, due to {the} orthogonal D2D and DL phases, each D2D content exchange between users adds extra time for delivering the content locally. 
%Since we follow a TDMA approach, the D2D transmissions occur in orthogonal time slots. Thus, each D2D transmission adds some amount to the total delivery time.
On the other hand, each successful D2D transmission reduces the remaining number of subfiles to be transmitted by the BS. Therefore, there are fewer multicast messages and corresponding beamforming vectors $\mathbf{w}_{\mathcal{D}}^{\mathcal{S}}$ in the DL optimization problem. As a result, a more efficient (less constrained) multicast beamformer design is possible, reducing the DL phase duration $T_{\mathrm{DL}}$. Therefore, the D2D mode selection is iteratively carried out as long as the following condition holds: 
\begin{equation} 
\label{eq:therishold}
  \frac{\hat{T}^{i}_{\mathrm{DL}}}{N_{s}^{i}} \geq \hat{T}_{\mathrm{D2D}}^i, \quad i\in\left[1, \ \binom{\tau+L}{\tau+1}\right],
  %N_{\mathrm{F}} -(\tau+1)(i-1)
\end{equation}
% \rev{I have added the general idea why I have defined such a criteria at the end of the paragraph}
where $N_{s}^{i} = (\tau+1)\Big(\binom{\tau+L}{\tau+1}-(i-1)\Big)$ {is the number of subfiles that is delivered in the DL phase assuming $i$ D2D subsets are selected.} Moreover, $\hat{T}^{i}_{\mathrm{DL}}$ and $\hat{T}_{\mathrm{D2D}}^i$ are the coarse approximated delivery times in the $i^\text{th}$ iteration.

%In~\eqref{eq:therishold}, we check if any D2D user subset will reduce the DL duration $T_{\mathrm{DL}}$ more than the corresponding D2D duration. Since we use an orthogonal time transmission scheme, we divide the total approximated DL time ($\hat{T}^{i}_{\mathrm{DL}}$) by the number of subfiles ($N_{s}^{i}$), assuming successive transmission of the subfiles. Thus, the average D2D transmission time for a single subfile must be less than the corresponding average DL transmission time. If a specific subset $\mathcal{D}$ in iteration $i$ satisfies~\eqref{eq:therishold}, then the D2D transmission for this subset is done following the approach proposed in \cite{Ji2016}, and we set $I_{\mathrm{D2D}}(\mathcal{D})=1$ for this subset.

In~\eqref{eq:therishold}, we check if any D2D user subset will reduce the DL duration $T_{\mathrm{DL}}$ more than the {added} D2D duration. Note that in each D2D time slot, $|\mathcal{D}| = \tau+1$ subfiles are delivered through $\tau+1$ orthogonal D2D transmissions. On the other hand, all the remaining subfiles (i.e., $N^{i}_{s}$ subfiles) are delivered simultaneously in the DL phase. Thus, in~\eqref{eq:therishold}, the average delivery time for a single subfile in the D2D and DL phases is compared. To this end, we divide the total approximated DL time ($\hat{T}^{i}_{\mathrm{DL}}$) by the number of subfiles ($N_{s}^{i}$), approximating the average transmission time for a single subfile. Thus, the average D2D transmission time {for a} subfile must be less than the corresponding average DL transmission time. In each iteration, we {set a subset as a D2D candidate}, i.e., the subset which provides the lowest delivery time. {If a specific subset $\mathcal{D}$ in iteration $i$ satisfies~\eqref{eq:therishold}, we set $I_{\mathrm{D2D}}(\mathcal{D})=1$ for this subset, and the D2D transmission is done as in~\cite{Ji2016}}. If~\eqref{eq:therishold} does not hold at any specific iteration, using more D2D transmissions will not improve the rate, and the iterative process is terminated. Therefore, at most, $\binom{\tau+L}{\tau+1}$ iterations are required for the proposed iterative process, while $2^{\binom{\tau+L}{\tau+1}}$ iterations are needed for the exhaustive search.
 
 The D2D delivery time is coarsely approximated as 
%%%%%%%%%%%%%%%%%%%%%%%%%%%%%%%%%%%%%%%%%%%%%%%%%%%%%%%% one column
%\begin{align}
%    &\hat{T}_{\mathrm{D2D}}^i  = \min_{\mathcal{D} \subseteq \Omega^{\mathcal{S}}}\hat{T}_{\mathrm{D2D}}^{\mathcal{D}}, \label{eq:Tselection}\ \\
%    &\hat{T}_{\mathrm{D2D}}^{\mathcal{D}} =\frac{1}{|\mathcal{D}|}\sum_{k \in \mathcal{D}}\frac{\frac{C(K,\tau,L)}{|\mathcal{D}|-1}}{R^{\mathcal{D}}_{k}}\label{Eq:apd2d}.
%\end{align} 
%%%%%%%%%%%%%%%%%%%%%%%%%%%%%%%%%%%%%%%%%%%%%%%%%%%%%
\begin{align}\label{eq:Tselection}
    &\hat{T}_{\mathrm{D2D}}^i  = \min_{\mathcal{D} \subseteq {\Omega^{\mathcal{S}}}}\frac{1}{|\mathcal{D}|}\sum_{k \in \mathcal{D}}\frac{\frac{C(K,\tau,L)}{|\mathcal{D}|-1}}{R^{\mathcal{D}}_{k}}.%\hat{T}_{\mathrm{D2D}}^{\mathcal{D}}, \ 
    %\hat{T}_{\mathrm{D2D}}^{\mathcal{D}} =\frac{1}{|\mathcal{D}|}\sum_{k \in \mathcal{D}}\frac{\frac{C(K,\tau,L)}{|\mathcal{D}|-1}}{R^{\mathcal{D}}_{k}}.%\label{Eq:apd2d}.
\end{align} 
In each D2D transmission (see Fig. \ref{fig:timeDiv}), $\frac{1}{|\mathcal{D}|-1}$ fractions of each subfile (of size $C(K,\tau,L)$) are delivered by user $k \in \mathcal{D}$ at the rate $R^{\mathcal{D}}_{k}$. Moreover, in each D2D subset, $|\mathcal{D}|=\tau+1$ subfiles are delivered. Thus, in \eqref{eq:Tselection}, the total required time is divided by $|\mathcal{D}|$ to compute the average {delivery} time for a single subfile. Note that for each iteration $i$, we only consider those subsets that have not yet been allocated for D2D transmissions.
 
To approximate $T_\text{DL}$, we make the following assumptions. First, we assume the beamformer $\mathbf{w}_\mathcal{D}^\mathcal{S}$ can remove all the interference for user $k \in \mathcal{D}$. Thus, we may assume $\mathcal{I}_{k} = \emptyset$, $\forall k \in \mathcal{S}$ in \eqref{Eq:MAC_general}. Second, we assume the beamformer $\mathbf{w}_\mathcal{D}^\mathcal{S}$ in \eqref{Eq:MAC_general} is matched to the channels of all the users in subset $\mathcal{D}$, i.e., the received SNR for $\hat{X}_\mathcal{D}^\mathcal{S}$ at receiver $k$ is $\text{SNR}^{k}_{\mathcal{D}}=\frac{\lVert \mathbf{h}_k \rVert^2\text{P}_{\mathcal{D}}}{N_{0}}$. Note that the beamformer  $\mathbf{w}_\mathcal{D}^\mathcal{S}$ is designed such that all the users in subset $\mathcal{D}$ can decode the message $\hat{X}_\mathcal{D}^\mathcal{S}$. Thus, to reflect this, we make use of the user’s channel gain for the heuristic mode selection process and limit the rate to the weakest user.\footnote{Another interpretation for \eqref{eq:aproximatedl} is that the beamformer $\mathbf{w}_\mathcal{D}^\mathcal{S}$ is assumed to be matched to the weakest user in subset $\mathcal{D}$ without rate loss for other users with better channel condition in the subset.}
Thus, the DL delivery time is coarsely approximated as
\begin{align} \label{eq:aproximatedl}
    &\hat{T}^{i}_{\mathrm{DL}} = \frac{C(K,\tau,L)}{{\min_{k\in [\mathcal{S}]}}
    \frac{1}{|\Omega_k^{\mathcal{S}}|}\log\left(1+\frac{\lVert \mathbf{h}_k \rVert^2 }{N_{0}}\sum\limits_{\mathcal{D}\subseteq{\Omega^{\mathcal{S}}_{k}}} \text{P}_{\mathcal{D}} \right)},%\quad
    %\hat{R}^{i}_{\mathrm{DL}} = {\min_{k\in [\mathcal{S}]}}\hat{R}^{i}_{k}, \nonumber \\
    %&\hat{R}^{i}_{k} = \min_{\mathcal{B} \subseteq \Omega_k^\mathcal{S}} \left[\frac{1}{|\mathcal{B}|}\log\left(1+\frac{|\mathcal{B}|\lVert \mathbf{h}_k \rVert^2\text{SNR}}{|\Omega^\mathcal{S}|} \right) \right],
    %&\hat{R}^{i}_{k} = \frac{1}{|\Omega_k^\mathcal{S}|}\log\left(1+\frac{|\Omega_k^\mathcal{S}|\lVert \mathbf{h}_k \rVert^2\text{SNR}}{|\Omega^\mathcal{S}|} \right),
    %&\hat{R}^{i}_{k} = \frac{1}{|\Omega_k^{\mathcal{S}}|}\log\left(1+\frac{\lVert \mathbf{h}_k \rVert^2 }{N_{0}}\sum\limits_{\mathcal{D}\subseteq{\Omega^{\mathcal{S}}_{k}}} \text{P}_{\mathcal{D}} \right),
\end{align}
where $\text{P}_{\mathcal{D}}$ is the dedicated power to the message $\hat{X}_\mathcal{D}^\mathcal{S}$. Denote $\hat{R}^{i}_{k} \triangleq \frac{1}{|\Omega_k^{\mathcal{S}}|}\log\left(1+\frac{\lVert \mathbf{h}_k \rVert^2 }{N_{0}}\sum\limits_{\mathcal{D}\subseteq{\Omega^{\mathcal{S}}_{k}}} \text{P}_{\mathcal{D}} \right)$ as the approximated rate of user $k$ assuming $i-1$ subsets have been chosen for D2D transmission. Note that $\hat{R}^{i}_{k}$ can be interpreted as the achievable rate of equivalent single-user {MISO} {MAC} channel with several (i.e., $|\Omega^{\mathcal{S}}_{k}|$) useful terms and no interference. %User $k \in \mathcal{D}$ receives the message $\hat{X}_\mathcal{D}^\mathcal{S}$ with SNR proportional to $\frac{\text{P}_{\mathcal{D}}}{N_0}$. 

To reflect the max-min objective in~\eqref{Eq:MulticastRate}, we assume the power {is divided among different messages such that} minimum received SNR for any two different messages are equal, %\footnote{This comes as the result of the rate expressions in the MAC region in~\eqref{Eq:MulticastRate}. Other power approximation approaches are possible using a different method, e.g., sum rate expressions in the MAC region.}, 
i.e.,
%\begin{align} \nonumber
    $\min_{i\in \mathcal{U}} \lVert \mathbf{h}_i \rVert^2\frac{\text{P}_{\mathcal{U}}}{N_0} = \min_{j\in \mathcal{D}} \lVert \mathbf{h}_j \rVert^2\frac{\text{P}_{\mathcal{D}}}{N_0}, \ \forall  \{\mathcal{D}, \ \mathcal{U}\}\in {\Omega^{\mathcal{S}}}, \ \mathcal{D} \neq \mathcal{U}$.
%\end{align}
Accordingly, the closed-form solution for $P_{\mathcal{D}}$ is given as follows
\begin{align} \label{eq:GeneralPower}
   \text{P}_{\mathcal{D}} = \frac{\prod\limits_{\mathcal{U}\subseteq{\Omega^{\mathcal{S}}}/\mathcal{D}} \min_{k\in \mathcal{U}} \lVert \mathbf{h}_k \rVert^2 P_T}{\sum\limits_{\mathcal{V}\subseteq{\Omega^{\mathcal{S}}}} \prod\limits_{\mathcal{U}\subseteq{\Omega^{\mathcal{S}}}/\mathcal{V}} \min_{i\in \mathcal{U}} \lVert \mathbf{h}_i \rVert^2}, \quad \forall \mathcal{D} \in \Omega^{\mathcal{S}}. 
\end{align}
%for simplicity we have assumed that the allocated powers to all the messages are equal (which is almost true for most of the cases), thus the received $\text{SINR}_{k}\propto \frac{\text{SNR}}{|\Omega^\mathcal{S}|}\propto\frac{P}{N_0|\Omega^\mathcal{S}|}$.
%Here, for simplicity, we omit the interference among parallel multicast streams and consider equal power loading over all the remaining subsets of users ($\frac{\text{SNR}}{|\Omega^\mathcal{S}|N_0}$). %Furthermore, for a given subset, the average channel gain is considered ($\frac{1}{|\mathcal{D}|} \sum_{j \in \mathcal{D}} \lVert \mathbf{h}_j \rVert^2$) to not overly emphasize specific channel conditions, but, rather, indicate the average available gains. % For mode selection, the average received SNR is sufficient indication of the relative performance in the DL. %Since, in DL, we are not limited by the available DoF, we can assume that users with similar channel conditions achieve rate proportional to the D2D phase with the same SNR. Note that this is not necessarily true for all channel conditions, but works well for the majority of cases as indicated by the simulation results in Section~\ref{sec:simres}. 
It is worth mentioning that when users experience similar channel conditions, the power allocated to each message can be assumed to be almost equal; therefore, \eqref{eq:aproximatedl} can be simplified to the approximated DL time in \cite{Mahmoodi-etal-Arxiv19}, i.e.,
\begin{equation}\nonumber
    \hat{T}^{i}_{\mathrm{DL}} \sim \frac{C(K,\tau,L)}{{\min_{k\in [\mathcal{S}]}}
    \frac{1}{|\Omega_k^\mathcal{S}|}\log\left(1+\frac{|\Omega_k^\mathcal{S}|\lVert \mathbf{h}_k \rVert^2P_T}{|\Omega^\mathcal{S}|N_0} \right)}.
\end{equation}

%Furthermore, $T_{D2D}^i$ is the approximate time that is needed to transfer one fragment of a file in D2D transmission considering that ($i-1$) number of D2D transmissions are done. Thus, we decide on using $i'\text{th}$ D2D transmission based on the assumption that ($i-1$) D2D transmission is performed and the total number of remaining fragments is ($\mathcal{M}_{total}-(t+1)(i-1)$). 
Once the user subsets for the D2D phase are selected, the final delivery time/rate is computed as in Section~\ref{sec:delivery_time_general}.
Note that when $I_{D2D}(\mathcal{D})=1$, the coded message corresponding to subset $\mathcal{D}$ is already delivered in the D2D phase. Thus, we can ignore such a subset in the DL phase, resulting in less inter-message interference and lower DL delivery time than~\cite{Tolli-Shariatpanahi-Kaleva-Khalaj-Arxiv18}.
Finally, the complete algorithm for the proposed two-phase delivery scheme is given in Algorithm~\ref{Alg_Main}.
{\begin{remark}
    The proposed D2D/DL mode selection in Algorithm~\ref{Alg_Main} is based on instantaneous channel knowledge and does not need any previous data history to approximate $\hat{T}_{\text{D2D}}$ or  $\hat{T}_{\text{DL}}$. However, the time approximation proposed in this work can be further improved by collecting statistics about users' channel conditions over a period and applying ML tools to approximate the D2D and DL transmission times. 
\end{remark}}
\begin{algorithm*}
    \small
	\caption{D2D Assisted Multi-Antenna Coded Caching\label{Alg_Main}}
	\begin{algorithmic}
		
		\Procedure{DELIVERY}{$W_1,\dots,W_N$, $d_1,\dots,d_K$, $\mathbf{H} = [\mathbf{h}_1, \dots, \mathbf{h}_K]$}
		\State $\tau \gets MK/N$
		
%		\ForAll{$\mathcal{D} \subseteq [K], |\mathcal{D}|=t+1$}
%		\State $I_{D2D}(\mathcal{D})=0 $
%		\EndFor		

		\For{$i\in\big[1, \binom{\tau+L}{\tau+1}\big]$} %\Comment{This is the beginning of D2D Phase.}
	%\State $I_{D2D}(\mathcal{D})=0 $
		\If{$\frac{\hat{T}^{i}_{\mathrm{DL}}}{N^{i}_{s}} \geq \hat{T}_{\mathrm{D2D}}^i$}
		\\
		\ForAll{$k \in \mathcal{D}$} %\Comment{Each loop pass is one D2D transmission. In each loop a subset $\mathcal{D}$ is selected based on \eqref{eq:Tselection}}.
		\\
		\State Each sub-file is divided into $\tau$ mini-file fragments.
		\State $X^\mathcal{D}_k\gets \oplus_{i \in \mathcal{D} \backslash \{k\}}  NEW(W_{{d_i},\mathcal{D}\backslash\{i\}}) $
		\State User $k$ multicasts $X^\mathcal{D}_k$ to $\mathcal{R}^{\mathcal{D}}(k)=\mathcal{D} \backslash \{k\}$ with the rate $R^{\mathcal{D}}_k$ stated in \eqref{eq:D2D_multicast_rate}
		\State $I_{D2D}(\mathcal{D})=1 $
		\EndFor	
		
				\EndIf
				\\
		\EndFor	%\Comment{This is the end of D2D Phase, which was based on the approach used in \cite{Ji2016}.}		
		
		\ForAll{$\mathcal{S} \subseteq [K], |\mathcal{S}|=\min(\tau+L,K)$} %\Comment{This is the beginning of DL phase.}

		\ForAll{$\mathcal{D} \subseteq \mathcal{S}, |\mathcal{D}|=\tau+1, I_{D2D}(\mathcal{D})=0$}

		\State $X_\mathcal{D}^\mathcal{S} \gets \oplus_{k \in \mathcal{D}}  NEW(W_{{d_k},\mathcal{D}\backslash\{k\}}) $
		\EndFor
		\State $\{\mathbf{w}_\mathcal{D}^\mathcal{S}\}=\arg \!\!\!\!\!\!\!\!\!\!\!\!\!\! \max\limits_{\{\mathbf{\underline{w}}_\mathcal{D}^\mathcal{S}, \mathcal{D} \subseteq \mathcal{S}, |\mathcal{D}|=\tau+1, I_{D2D}(\mathcal{D})=0\}} \!\!\!\!\!\!\!\!\!\!\!\!\!\! R_C\left(\mathcal{S},\{\underline{\mathbf{w}}_\mathcal{D}^\mathcal{S}, \mathcal{D} \subseteq \mathcal{S}, |\mathcal{D}|=\tau+1, I_{D2D}(\mathcal{D})=0\}\right)$ %\Comment{$R_C$ is defined in \eqref{Eq:MulticastRate}.}
		\State $\underline{\mathbf{X}}(\mathcal{S}) \gets \sum_{\mathcal{D} \subseteq \mathcal{S}, |\mathcal{D}|=\tau+1, I_{D2D}(\mathcal{D})=0} \mathbf{w}_\mathcal{D}^\mathcal{S} \tilde{X}_\mathcal{D}^\mathcal{S}$

		\State \textbf{transmit} $\underline{\mathbf{X}}(\mathcal{S})$ with the rate $R_C(\mathcal{S},\{\mathbf{w}_\mathcal{D}^\mathcal{S}, \mathcal{D} \subseteq \mathcal{S}, |\mathcal{D}|=\tau+1, I_{D2D}(\mathcal{D})=0\})$.

		\EndFor %\Comment{This is the end of DL phase, which was based on the approach used in \cite{Toll1806:Multicast}.}
		\EndProcedure
		
	\end{algorithmic}
\end{algorithm*}
%%%%%%%%%%%%%%%%%%%%%%%%%%%%%%%%%%%%%%%%%%%%%%%%%%%%%%%%%%%%%%%%%%%%
\subsection{Heuristic D2D mode selection for restricted DoF}\label{sec:restrictedDoF}
%%%%%%%%%%%%%%%%%%%%%%%%%%%%%%%%%%%%%%%%%%%%%%%%%%%%%%%%%%%%%%%%%%%%
The proposed iterative D2D mode selection can be extended to the system setting with restricted DoF~\cite{Tolli-Shariatpanahi-Kaleva-Khalaj-Arxiv18}. The authors in~\cite{Tolli-Shariatpanahi-Kaleva-Khalaj-Arxiv18} propose limiting the DoF by serving $\tau+\alpha$ ($\alpha \leq L$) users at each transmission phase, resulting in a less complex beamformer design. Furthermore, they divide the users into $P$ distinct groups (for some $P \in \mathbb{N}$) to decrease the number of overlapping groups (c.f.~\cite{Tolli-Shariatpanahi-Kaleva-Khalaj-Arxiv18}). The combination of D2D and DL transmissions proposed in this paper also applies to the $\alpha < L$ case. The difference is that the total number of different D2D subsets changes from $\binom{\tau+L}{\tau+1}$ to $P\binom{\tau+\beta}{\tau+1},  ${where $\beta\triangleq\frac{\tau+\alpha}{P}-\tau$ is an integer}. Accordingly, $\Omega^{{\mathcal{S}}}$,  $\Omega^{{\mathcal{S}}}_{k}$, and $C(K,\tau,L)$ change to the ones defined in~\cite{Tolli-Shariatpanahi-Kaleva-Khalaj-Arxiv18}, and the process remains the same. This paper is a particular case of the system proposed in~\cite{Tolli-Shariatpanahi-Kaleva-Khalaj-Arxiv18} where $P=1, \text{and} \  \alpha= \beta = L$.
%%%%%%%%%%%%%%%%%%%%%%%%%%%%%%%%%%%%%%%%%%%%%%%%%%%%%%%%%%%%%%%%%%
\subsection{D2D aided beamforming for general group sizes} \label{sec: D2D-general-groupsize-LTP1}
%%%%%%%%%%%%%%%%%%%%%%%%%%%%%%%%%%%%%%%%%%%%%%%%%%%%%%%%%%%%%%%%%%%%
%So far we have considered group size equal to $\tau +1$ as in~\cite{Ji2016}. It was shown in~\cite{D2D-CC-Optload-memtradeof-caire-2019} that, under constant link capacity and error free link model (considered in~\cite{Ji2016}), the proposed D2D method in~\cite{Ji2016} is within a constant factor of any optimal D2D transmission. As simulation results show, when users experience similar channel conditions, considering D2D group size less than $\tau+1$ does not effect the total delivery time significantly (agreeing the results in~\cite{D2D-CC-Optload-memtradeof-caire-2019}). However, when users experience different D2D link capacity, the results in~\cite{D2D-CC-Optload-memtradeof-caire-2019}
%does not hold anymore and considering group sizes less that $\tau+1$ can have a great impact on the total delivery time. 
In this section, we extend the results in sections~\ref{sec: general-parta} and \ref{section: general-partB} to general D2D group sizes. %Considering group size $|\mathcal{V}| = v$ ($2 \leq v < \tau +1$),
Let us define a new set $\overline{\Omega^{\mathcal{S}, \mathcal{V}}} := \{\mathcal{V} \subseteq \mathcal{S}, 2 \leq |\mathcal{V}| \leq \tau+1, I_{\text{D2D}}(\mathcal{V}) = 1 \}$ as the set of D2D groups (of any size) selected for D2D phase. % and 2) $\Omega^{\mathcal{S}, \mathcal{V}}_{k} := \{ \mathcal{D} \subseteq \mathcal{S} :  \ |\mathcal{D}| = \tau+1, \ I_{\text{D2D}}(\mathcal{D})=0 \ | \ \forall \mathcal{V} \subset \mathcal{S}, \ I_{\text{D2D}}(\mathcal{V})  = 0, \ 2 \leq |\mathcal{V}| < \tau+1, \ k \in \mathcal{V} \}$ is the set of the message indices needed at user $k$ in the DL phase.
Now, for a given D2D subset selection, the $T_{\mathrm{D2D}}$ is computed as the following
\begin{align}
    &T_{\mathrm{D2D}}  = \sum_{\mathcal{V} \subseteq \overline{\Omega^{\mathcal{S},\mathcal{V}}}}\sum_{k \in \mathcal{V}}\frac{\frac{a^{\mathcal{V}}_{k}C(K,\tau,L)}{(|\mathcal{V}|-1)}}{{R}^{\mathcal{V}}_{k}}\label{Eq:td2dv},
\end{align}
where $a^{\mathcal{V}}_{k}$ is the number of transmitted messages by user $k \in \mathcal{V}$. Please note that though $a^{\mathcal{V}}_{k} = 1$ for $|\mathcal{V}| = \tau+1$, $a^{\mathcal{V}}_{k}$ can be any number for $|\mathcal{V}| < \tau+1$. The corresponding DL rate ($R_C$) for the general group sizes is computed using \eqref{Eq:MulticastRate}. %and substituting $\Omega^{\mathcal{S}}_{k}$ with $\Omega^{\mathcal{S}, \mathcal{V}}_{k}$ ($\mathcal{I}_{k}$ is defined as before). 
The heuristic mode selection criteria~\eqref{eq:therishold} changes as follows
\begin{equation} 
\label{eq:general-c-theri}
  \frac{\hat{T}^{i}_{\mathrm{DL}}}{N_{v}^{i}} \geq \hat{T}_{\mathrm{D2D}}^i, \quad i\in\left[1, \ \sum_{j=2}^{\tau + 1}\binom{\tau+L}{j}\right],
\end{equation}
where $N_{v}^{i}$ is the total number of remaining subfiles delivered through DL transmission in the $i$'th iteration. Here, we first check the subsets of size $|\mathcal{V}|=\tau + 1$ using~\eqref{eq:Tselection}; after all the subsets of size $|\mathcal{V}|=\tau + 1$ are checked, we continue the procedure for $2 \leq |\mathcal{V}|< \tau +1$. We use the same approximation for $\hat{T}^{i}_{\mathrm{DL}}$ and $\hat{T}_{\mathrm{D2D}}^i$ given in~\eqref{eq:Tselection} and~\eqref{eq:aproximatedl}, respectively. %The main difference to the original method is that now $|\mathcal{D}|\neq \tau+1$ and $\Omega^{\mathcal{S}}_{k}/\overline{\Omega^{\mathcal{S}}}$ must be substituted by $\Omega^{\mathcal{S}, \mathcal{V}}_{k}/\overline{\Omega^{\mathcal{S}, \mathcal{V}}}$. 
{However}, for general group sizes, each transmitted message $\hat{X}_\mathcal{D}^\mathcal{S}$ in the DL phase may not be useful for all $ k \in \mathcal{D}$ (some users in $\mathcal{D}$ may have received the transmitted useful term in D2D mode). Thus, in~\eqref{eq:GeneralPower}, the minimum is taken over those users who still need the message.

Similar to section~\ref{section: general-partB}, the D2D subset selection is carried out as long as \eqref{eq:general-c-theri} holds. The difference is that for group size less than $\tau+1$, each user may have several contents to transmit to the other users (i.e., $a^{\mathcal{V}}_{k} \neq 1$). When a subset $\mathcal{V}$ is selected for D2D transmission, users in $\mathcal{V}$ transmit all the useful data available in their cache to the other users in the selected subset (following the method in~\cite{Ji2016}) and $I_{\mathrm{D2D}}(\mathcal{V})$ is set to one. We also set $I_{\mathrm{D2D}}(\mathcal{D})$ to one if all the subfiles in $\mathcal{D}$ are transmitted through D2D groups of size $\mathcal{V} < \tau+1$. The rest of the process is the same as in Section~\ref{section: general-partB}. After D2D subset assessment is done, the DL delivery time is computed using~\eqref{eq:aproximatedl}, % where $\Omega^{\mathcal{S}}_{k}$ is substituted with $\Omega^{\mathcal{S}, \mathcal{V}}_{k}$ in the computation of $R_C$ in~\eqref{Eq:MulticastRate}. 
and the D2D time is computed using~\eqref{Eq:td2dv}.

{It is worth mentioning that the main target in~\eqref{eq:therishold} and~\eqref{eq:general-c-theri} is to reduce the overall transmission time (i.e., $T_{\text{D2D}}+T_{\text{DL}}$) in~\eqref{eq:total_rate}. However, as we show in the following section, D2D transmission can also notably reduce the complexity of the beamforming process in the DL phase. To this end, we briefly illustrate the complexity involved in the beamformer design and the effect of D2D transmission on the process.}
%%%%%%%%%%%%%%%%%%%%%%%%%%%%%%%%%%%%%%%%%%%%%%%%%%%%%%%%%%%%%%%%%%%%%%%%%
\section{beamforming complexity analysis}
\label{sec:complexity_analyze}
%%%%%%%%%%%%%%%%%%%%%%%%%%%%%%%%%%%%%%%%%%%%%%%%%%%%%%%%%%%%%%%%%%%%%%%%%

In this section, we investigate the effects of D2D transmissions in computational complexity for the general case. Authors in~\cite{Tolli-Shariatpanahi-Kaleva-Khalaj-Arxiv18} show that the number of MAC conditions and quadratic terms in the SINR constraints dominates the complexity of the DL beamformer design. Thus, we first introduce two boundaries for the number of MAC conditions, then discuss the effects of D2D on the beamformer design complexity.

\begin{theorem}
Considering $i(\tau+1)+m$ subfiles are delivered {via D2D transmissions}, among which $i(\tau+1)$ subfiles are delivered through D2D groups with size $(\tau+1)$, the number of MAC conditions for the DL phase is lower bounded by:
%\vspace{-10px}
%%%%%%%%%%%%%%%%%%%%%%%%%%%%%%%%%%%%%%%%%%%%%%%%%%%%%%%%%%%one column
%\begin{align}
%    \text{MAC}^{i}_\text{min} =&(\tau+L-b)(2^a-1)+b(2^{a+1}-1),&\label{eq:NA3}
%    \end{align}
%    where
%\begin{align}
%    a=&\left\lfloor{\frac{(\tau+1)\left(\binom{\tau+L}{\tau+1}-i\right)-m}{\tau+L}}\right\rfloor,\quad  \text{for} \ \ i \leqslant \binom{\tau+L}{\tau+1},& \label{eq:NA1} \\
%    b=&(\tau+1)(\binom{\tau+L}{\tau+1}-i)-a(\tau+L),&\label{eq:NA2}\\
%    & m = \sum_{\mathcal{V} \in \overline{\Omega^{\mathcal{S}, \mathcal{V}}} \setminus \overline{\Omega^{\mathcal{S}}}}\sum_{k \in \mathcal{V}}a^{\mathcal{V}}_{k} \label{eq:tot-subfile-LTP1},
%    \end{align}
%    and the maximum number of {MAC} conditions for the {DL} phase is
%\begin{align}
%    \text{MAC}^{i}_\text{max}=&(\tau+L-U)(2^{\binom{\tau+L-1}{\tau}}-1) + (2^{\binom{\tau+L-1}{\tau}-X}-1) \nonumber\\ \hspace{1cm}&+(U-(U_{1}+1))(2^{(\binom{\tau+L-1}{\tau}-\binom{U-2}{\tau})}-1)   \nonumber \\
%    \hspace{1cm} &+(U_{1}-\phi)(2^{(\binom{\tau+L-1}{\tau}-(\binom{U-2}{\tau}+\binom{U_{1}-1}{\tau-1}-Y))}-1)\text{,} \label{eq:NA6}
%    \end{align}
%    where
%\begin{align}
%    X=&i  -\binom{U-1}{\tau+1}, \ \text{for}\ \binom{U-1}{\tau+1} <  i\leqslant \binom{U}{\tau+1}\nonumber,&\\ &U\leqslant (\tau+L),&\label{eq:NA4} \\
%    &\binom{U_{1}-1}{\tau} <  X\leqslant \binom{U_{1}}{\tau},& \label{eq:u1} \\
%    Y=&\binom{U_1}{\tau}-X, \ \phi = \left\lfloor{\frac{m}{\binom{\tau+L-1}{\tau}}}\right\rfloor,& \label{eq:NA5} 
%\end{align}
%%%%%%%%%%%%%%%%%%%%%%%%%%%%%%%%%%%%%%%%%%%%%%%%%%%%%%%%
\begin{equation} 
    \underline{M}(\tau,i,m,L) =(\tau+L-b)(2^a-1)+b(2^{a+1}-1),\label{eq:NA3}
\end{equation}
    where
\begin{align}
    a&\triangleq\left\lfloor{\frac{(\tau+1)\left({M}_{\text{T}}-i\right)-m}{\tau+L}}\right\rfloor,\quad
    \\ b&\triangleq\left(\tau+1\right)\left({M}_{\text{T}}-i\right)-m-a\left(\tau+L\right),\label{eq:NA1} \ \ %i \leqslant \binom{\tau+L}{\tau+1}, 
    \\
    m &\triangleq \sum_{\mathcal{V} \in \overline{\Omega^{\mathcal{S}, \mathcal{V}}} \setminus \overline{\Omega^{\mathcal{S}}}}\sum_{k \in \mathcal{V}}a^{\mathcal{V}}_{k}, \quad {M}_{\text{T}}  \triangleq \binom{\tau+L}{\tau+1} \label{eq:tot-subfile-LTP1},
    \end{align}
    and the number of MAC conditions is upper bounded by
\begin{align}
    & \overline{M}(\tau,i,m,L) =\nonumber\\ & \quad (\tau+L-U)(2^{W}-1) + (2^{W-X}-1) \nonumber\\ & \quad \quad + (U-(\phi+1))(2^{(W-\binom{U-2}{\tau})}-1) \text{,}  %\\ &+ (U_{1}-\phi)(2^{(\binom{\tau+L-1}{\tau}-(\binom{U-2}{\tau}+\binom{U_{1}-1}{\tau-1}-Y))}-1) + (2^{\binom{\tau+L-1}{\tau}-X}-1)\text{,} \label{eq:NA6}
    \end{align}
    where
\begin{align}
    &U \triangleq \min_{\text{S.T.} \ \binom{U^{'}}{\tau+1} \ \geq \ i} {U^{'}}, \ X\triangleq i  -\binom{U-1}{\tau+1} %\binom{U-1}{\tau+1} <  i\leqslant \binom{U}{\tau+1}
    , \\ & \phi \triangleq \left\lfloor{\frac{m}{W- \binom{U-2}{\tau}}}\right\rfloor, \ W \triangleq \binom{\tau+L-1}{\tau}.\label{eq:NA4} %\\
    %&Y\triangleq\binom{U_1}{\tau}-X, \quad
    %\binom{U_{1}-1}{\tau} <  X\leqslant \binom{U_{1}}{\tau},\quad \phi \triangleq \left\lfloor{\frac{m}{\binom{\tau+L-1}{\tau}}}\right\rfloor. \label{eq:NA5} 
\end{align}
%and where $\tau=\frac{KM}{N}$ and $M$ is the normalized user cache size. 

\end{theorem}
\begin{proof}
Refer to Appendix~\ref{sec:app}.
\end{proof}

The number of MAC conditions varies between $\underline{M}(.)$ and $\overline{M}(.)$ based on which particular subsets have been selected for the D2D phase. For better intuition consider Fig.~\ref{fig:NA1}, which shows the normalized maximum and minimum number of MAC conditions ($K=10,\ L=9,\ \tau=1$) for different {number of D2D groups} $i$. As shown, the number of MAC conditions decreases drastically by using just a few D2D transmissions, which in turn dramatically reduces the complexity of the DL beamformer design. For example, for the case depicted in Fig.~\ref{fig:NA1}, by choosing only five different subsets of users among 45 available subsets, the number of MAC conditions can be reduced to half. Therefore, D2D significantly improves the beamformer design complexity.

\begin{comment}

\begin{figure}
    \centering 
    \setlength\abovecaptionskip{-0.25\baselineskip}
    \includegraphics[width=0.5\columnwidth,keepaspectratio]{figs/MCterm.pdf}
    \caption{The normalized number of MAC conditions vs the number of subsets assigned for D2D transmissions with $k=10, \ L=9$ and $\tau=1$.}
\label{fig:NA1}
\end{figure}

\end{comment}

\begin{figure}[t!]
\begin{minipage}[c]{0.49\textwidth}
\centering
\setlength\abovecaptionskip{-0.25\baselineskip}
\includegraphics[width=\linewidth]{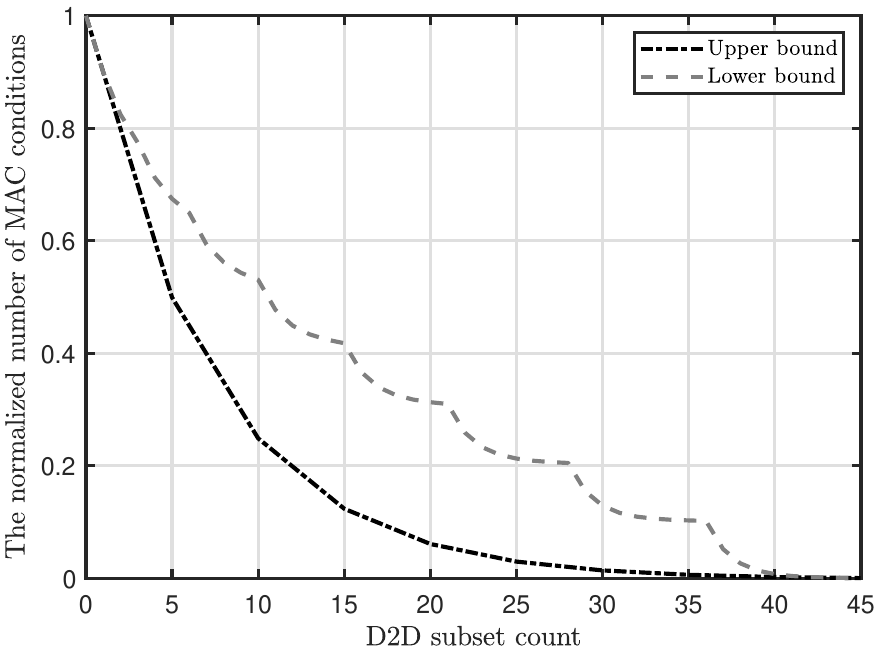}
\caption{The normalized number of MAC conditions vs the number of subsets assigned for D2D transmissions with $K=10, \ L=9$ and $\tau=1$.}
\label{fig:NA1}
\end{minipage}
\hspace{1mm}
\begin{minipage}[c]{0.49\textwidth}
\centering
\setlength\abovecaptionskip{-0.25\baselineskip}
\includegraphics[width=\linewidth]{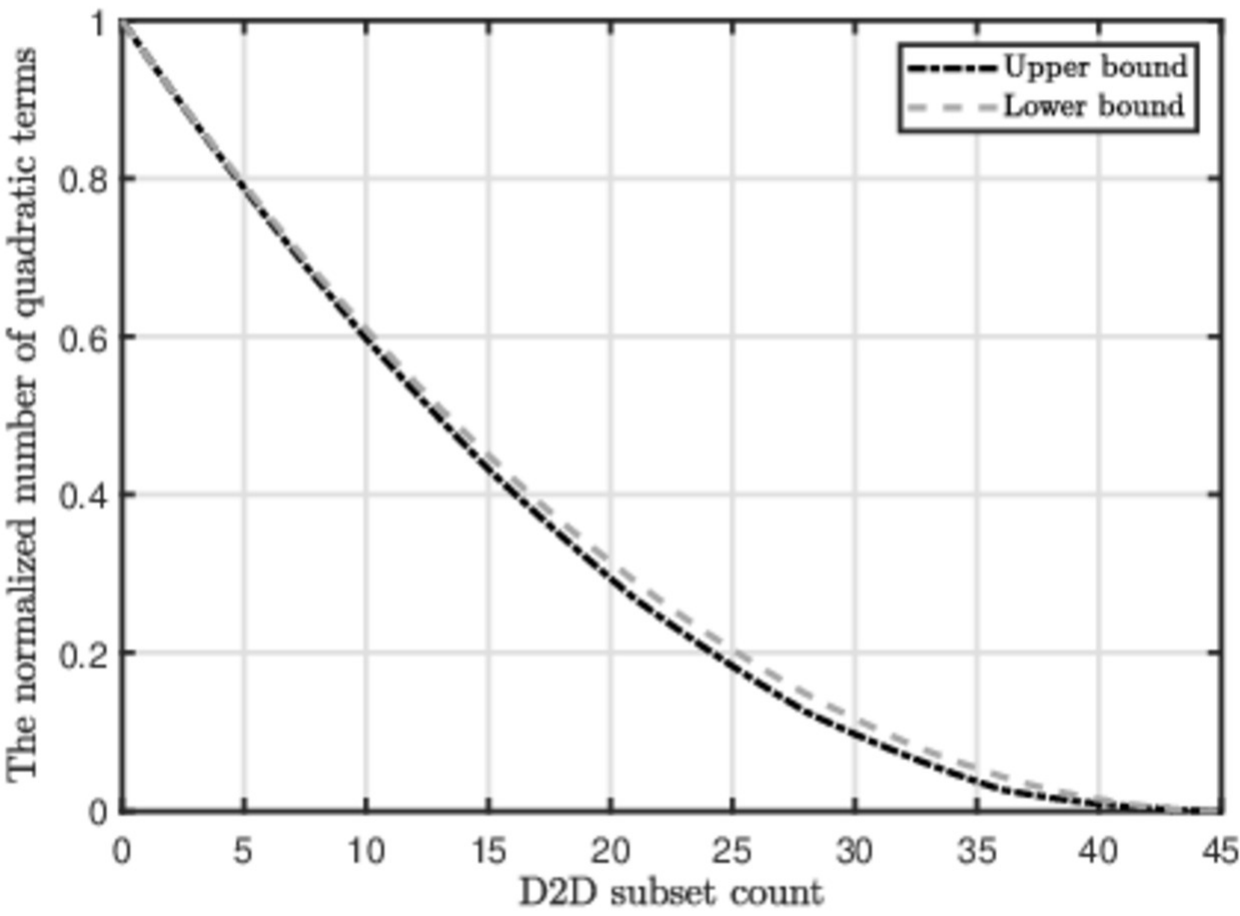}
\caption{The normalized number of quadratic variables vs the number of subsets assigned for D2D transmissions with $K=10, \ L=9$ and $\tau=1$.}
\label{fig:NA2}
\end{minipage} 
\end{figure}
 In the following, we provide the boundaries for the number of quadratic terms in the {SINR} constraints (the second important factor in the DL beamforming complexity).
\begin{theorem}
Considering $i(\tau+1)+m$ subfiles are delivered {via D2D transmissions}, among which $i(\tau+1)$ subfiles are delivered through D2D groups with size $(\tau+1)$, the number of quadratic terms is upper bounded by:
\begin{align}
    &\overline{Q}(\tau,i,m,L)= bA_2B_2+(\tau+L-b)A_1B_1,\label{eq:NA9}
\end{align}
where
%\begin{align}
    $A_1 \triangleq a,
    \ A_2 \triangleq a+1, \ B_1 \triangleq M_\text{T}-i-A_1+1,
    \ B_2 \triangleq M_\text{T}-i-A_2+1$. %\nonumber %\\ %\label{eq:NA8}\\
    %&\underline{M_\text{T}} \triangleq  \left\lceil \frac{\big(\tau+1\big)\Big(\binom{\tau+L}{\tau+1}-i\Big)-m}{\tau+1}\right\rceil, \nonumber
%\end{align}
Moreover, the number of quadratic terms is lower bounded by: 
\begin{align}
%\label{eq:NA7}\\
    \underline{Q}(\tau,i,m,L)= & (\tau+L-U)A_{1}^{'}B_{1}^{'} \nonumber\\ & \ \ + (U-(\phi+1))A_{2}^{'}B_{2}^{'}\
    %+(U_1-\phi)A_{3}^{'}B_{3}^{'} 
    + A_{3}^{'}B_{3}^{'}, \label{eq:NA11}
\end{align}
where
%\begin{align}
% \\
    $A_{1}^{'} \triangleq W,\ A_{2}^{'} \triangleq A_1^{'}-\binom{U-2}{\tau}, 
    %\ A_{3}^{'} \triangleq A_{2}^{'}-\binom{U_1-1}{\tau-1}+Y, 
    \ A_{3}^{'} \triangleq A_1^{'}-X, %\nonumber \\ 
    B_{1}^{'} \triangleq \underline{M_\text{T}}-A_{1}^{'}+1,\ B_{2}^{'} \triangleq \underline{M_\text{T}}-A_{2}^{'}+1, \ %B_{3}^{'} \triangleq \underline{M_\text{T}}-A_{3}^{'}+1, \ %\label{eq:NA10}
    B_{3}^{'} \triangleq \underline{M_\text{T}}-A_{3}^{'}+1$. %\nonumber %\\
    %&\overline{M_\text{T}} \triangleq \binom{\tau+L}{\tau+1}-i, \nonumber
%\end{align}
%where $Q^{i}_\text{max}$ and $Q^{i}_\text{min}$ are the maximum and minimum number of quadratic terms, respectively.  
We denote $\underline{M_\text{T}} \triangleq  \left\lceil \frac{\big(\tau+1\big)\left(M_\text{T}-i\right)-m}{\tau+1}\right\rceil$ as the lower approximation of the total number of messages sent by the BS. % after $i$ number of D2D time slots are done. 
Moreover, $a$, $b$, $X$, $U$, $m$, $M_\text{T}$, $W$ and $\phi$ are defined trough~\eqref{eq:NA1} to ~\eqref{eq:NA4}. %Finally, $K$ is the number of users and $i$ is the number of subsets that have been used in D2D transmission.
\end{theorem}
\begin{proof}
Refer to Appendix~\ref{sec:appB}.
\end{proof}
\begin{comment}

\begin{figure}
    \centering 
    \setlength\abovecaptionskip{-0.25\baselineskip}
    \includegraphics[width=0.5\columnwidth,keepaspectratio]{figs/Qterm.pdf}
    \caption{The normalized number of quadratic variables vs the number of subsets assigned for D2D transmissions with $k=10, \ L=9$ and $\tau=1$.} % The green curve here corresponds to the red curve in Fig.~\ref{fig:NA1} and the red one here corresponds to the green one in Fig.~\ref{fig:NA1}.}
\label{fig:NA2}
\end{figure}

\end{comment}

Fig.~\ref{fig:NA2} depicts the upper and lower boundaries for the same scenario as in Fig.~\ref{fig:NA1}. The gap between these bounds is not as {considerable} as MAC conditions. Thus, compared to the quadratic terms, the number of MAC conditions is more affected by how different D2D subsets are selected for transmission. Nevertheless, the role of D2D transmissions in reducing the total number of quadratic terms is notable. For example, choosing five different D2D subsets for the considered case reduces the total number of quadratic terms by $20 \%$. 

\begin{remark}
For ease of exposition,  $\tau+L = K$ is assumed in the equations and algorithms throughout this paper. However, the proposed methods can be easily generalized to other regimes. For example, in case $\tau+L <K$, there are $\binom{K}{\tau+L}$ orthogonal transmission phases. All the equations in this paper are valid separately for each phase. Similarly, for the case $\tau+L > K$, $\tau + L$ should be replaced with $K$ in all the equations. Moreover, for the restricted spatial DoF scenario~\cite{Tolli-Shariatpanahi-Kaleva-Khalaj-Arxiv18}, discussed also in Section~\ref{sec:restrictedDoF}, $\tau+L$ should be changed to $\tau+\alpha$. Finally, the $\beta$ parameter introduced in~\cite{Tolli-Shariatpanahi-Kaleva-Khalaj-Arxiv18} can be easily applied to the proposed equations by changing $M_\text{T}$ to $P\binom{\tau+\beta}{\tau+1}$, and $W$ to $\binom{\tau+\beta-1}{\tau}$, where $\alpha, \ \beta$ and $P$ are defined in Section~\ref{sec:restrictedDoF}.
\end{remark}
%%%%%%%%%%%%%%%%%%%%%%%%%%%%%%%%%%%%%%%%%%%%%%%%%%%%%%%%%%%%%%%%%%
\section{Numerical results}
\label{sec:simres}
%%%%%%%%%%%%%%%%%%%%%%%%%%%%%%%%%%%%%%%%%%%%%%%%%%%%%%%%%%%%%%%%%%
In this section, we provide numerical examples for two scenarios with $K=3$ and $K=4$ users (see Fig.~\ref{fig:sysmodel} and Fig.~\ref{fig:sysmodel_complex}). Due to the complex beamforming procedure for the multiserver-based schemes (such as the one proposed in this paper), we have considered a limited number of users in the network. We consider a circular cell with a radius of  $R=100$ meters, where the BS is located in the cell center. To investigate the effect of D2D transmission in different situations, we introduce a smaller circle with radius $r$ within the cell area, wherein the users are randomly scattered. Thus, the maximum distance between any two users is $2r$. In contrast, the  users' distance to BS varies between 0 and $R$. {Hence, by} changing $r$, the maximum users' separation in D2D mode is controlled, which helps us determine the beneficial users' distance in the D2D phase. 

For {D2D} transmissions, the channel gains are generated as ${h}_{ik} = d_{ik}^{-\frac{n_{\text{D2D}}}{2}}g_{ik}$, where $g_{ik} \sim \mathbb{CN}(0,1)$, $n_{\text{D2D}} = 2$ is the path-loss exponent, and $d_{ik}$ is the inter-user distance. The channel vectors for {DL} transmission are generated from i.i.d statistics with $\mathbf{h}_k = d_{k}^{-\frac{n_{\text{DL}}}{2}} {\mathbf{g}}_k$, where ${\mathbf{g}}_k \sim \mathbb{CN}(0,\text{$\mathbf{I}$})$, $n_{\text{DL}} = 3$ is the path-loss exponent, and $d_k$ is the BS-user distance. Transmit powers for D2D transmissions at the user side are adjusted so that the average received SNR at the receiver is $0$ dB at a $10$-meter distance. The BS transmit power is adjusted such that the average received SNR is $0$ dB at a $100$-meter distance. For comparison, we have also consider two benchmark schemes~\cite{Tolli-Shariatpanahi-Kaleva-Khalaj-Arxiv18} and~\cite{Ji2016} denoted as \textit{Multicasting only} and \textit{D2D~only}, respectively. {Note that the optimality of the scheme in~\cite{Ji2016} is shown in~\cite{D2D-CC-Optload-memtradeof-caire-2019}, and the superiority of the proposed method in~~\cite{Tolli-Shariatpanahi-Kaleva-Khalaj-Arxiv18} over the traditional unicasting approaches has been shown therein. Hence, we compare the proposed D2D/DL delivery scheme to these two schemes in this work.}

Fig.~\ref{fig:NA3} shows average delivery rate for $K=3$, $L=2$ and $\tau=1$ (section~\ref{sec:Example1}) as a function of inner circle radius $r$.
The figure demonstrates that when users are close to each other, there is a significant gain from using a combination of multicasting and D2D transmissions. Compared to the \textit{D2D only} rate that decreases drastically as the inter-user distance is increased, the proposed approach shows steadier behavior.
The beneficial range for D2D transmission in this particular scenario appears to be between $r= 0 \ \text{and} \ 5 m$ ($10 m$ maximum distance). The range can change if the path-loss exponent, the available power for both D2D and DL transmission, $\tau$, etc., are varied. The simulation results demonstrate that sending all the data only {via} D2D transmissions or only through multicasting results in a lower rate than the proposed approach with the optimized mode selection.
\begin{comment}

\begin{figure}
    \centering 
    \setlength\abovecaptionskip{-0.25\baselineskip}
    \includegraphics[width=0.5\columnwidth,keepaspectratio]{Fig/newrsltk3.pdf}
    \caption{Per-user rate vs. inner circle radius $r$ for $K=3$, $L=2$ and $\tau=1$.}
\label{fig:NA3}
\end{figure}

\begin{figure}
    \centering 
    \setlength\abovecaptionskip{-0.25\baselineskip}
    \includegraphics[width=0.5\columnwidth,keepaspectratio]{Fig/newrsltk4.eps}
    \caption{Per-user rate vs. inner circle radius $r$ for $K=4$, $L=2$ and $\tau=2$.}
\label{fig:NA4}
\end{figure}

\end{comment}

\begin{figure}[t!]
\begin{minipage}[c]{0.49\textwidth}
\centering
\setlength\abovecaptionskip{-0.25\baselineskip}
\includegraphics[width=\linewidth]{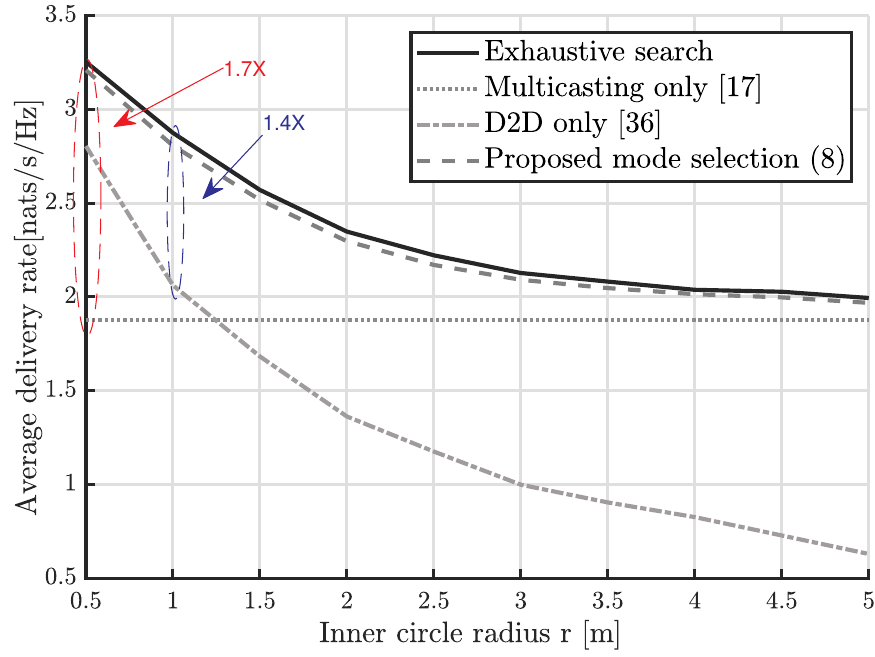}
\caption{Average delivery rate vs. inner circle radius $r$ for $K=3$, $L=2$ and $\tau=1$.}
\label{fig:NA3}
\end{minipage}
\hspace{1mm}
\begin{minipage}[c]{0.49\textwidth}
\centering
\setlength\abovecaptionskip{-0.25\baselineskip}
\includegraphics[width=\linewidth]{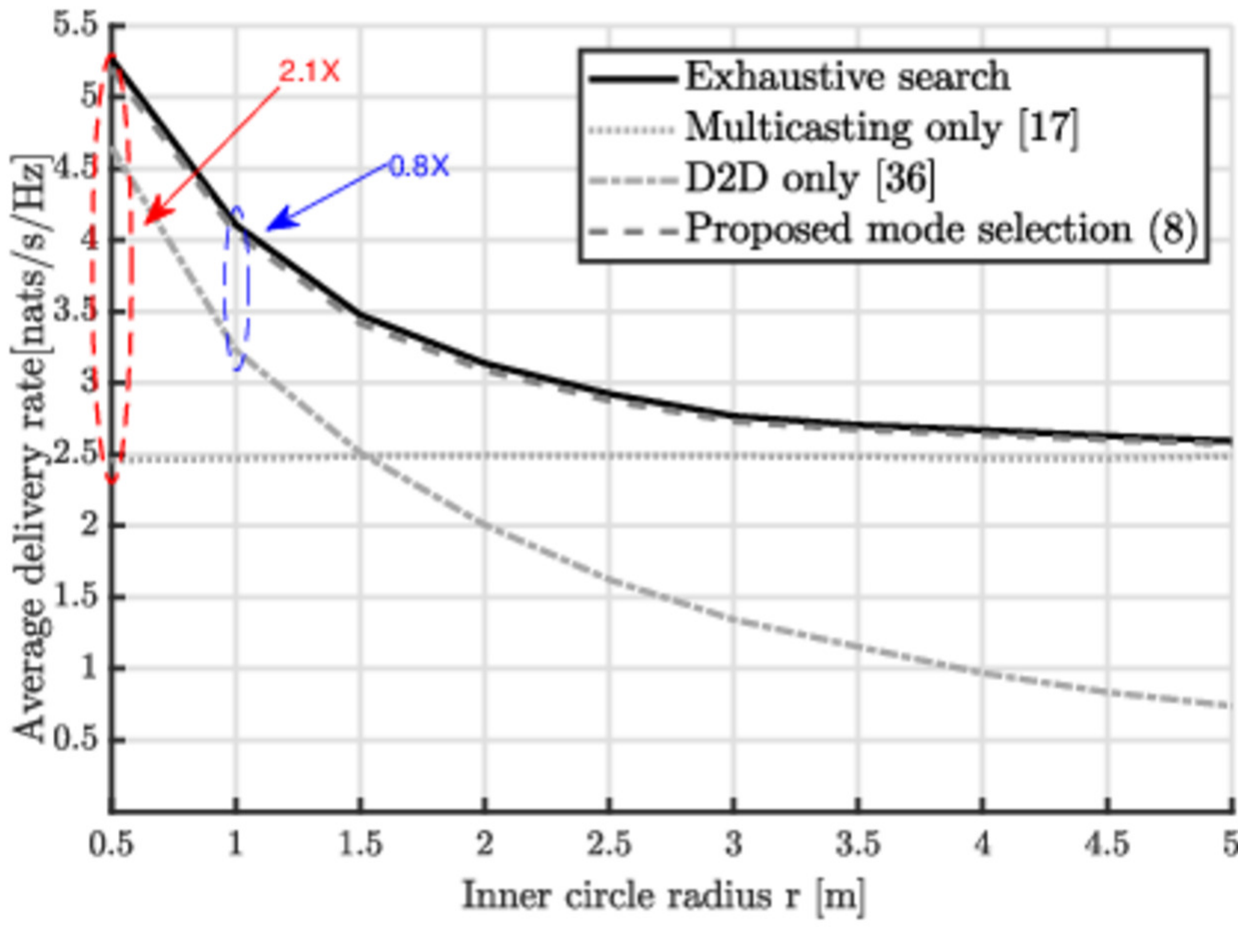}
\caption{Average delivery rate vs. inner circle radius $r$ for $K=4$, $L=2$ and $\tau=2$.}
\label{fig:NA4}
\end{minipage} 
\end{figure}

Fig.~\ref{fig:NA4} shows the average delivery rate versus the inner circle radius for $K=4,\ \tau=2,\  L=2$ (section~\ref{sec:Example2}). In this case, the gain from D2D transmission among nearby users is larger than the case $K=3$ due to more D2D transmission opportunities. However, the gain of D2D transmission decreases more rapidly compared to the case $K=3$. In this case, $\tau$ equals $2$, {so} more users {must} be closer to each other to perform efficient D2D transmission. Therefore, it can be concluded that for a fixed number of users $K$, increasing $\tau$ results in fewer D2D opportunities and fewer D2D subset selection variants. %For example, if $\mathcal{D}({1})= \{1,2,3\} \ \text{and} \ \mathcal{D}({2})=\{1,2,4\}$ are selected for the simulated scenario, most probably four users are in a close distance; therefore, $\mathcal{D}({3})=\{1,3,4\}$ and $\mathcal{D}({4})=\{2,3,4\}$ will also be selected with a high probability. 
However, with a fixed $\tau$, increasing the number of users $K$ will result in a more diverse D2D combination and higher gain over the \textit{D2D only} case.

It is worth mentioning that using the heuristic D2D mode selection criteria (defined in Section \ref{sec:general}) results in a minimal average delivery rate loss, with a significantly reduced complexity compared to the exhaustive search. Simulation results show that the approximated rate in \eqref{eq:aproximatedl} is very close to the actual rate \eqref{Eq:MAC_general} for different $\Omega^\mathcal{S}$ and different network parameters (i.e., $\tau,\ L,\ K,\ \text{etc}$).
\begin{figure}[t!]
    \centering 
    \setlength\abovecaptionskip{-0.25\baselineskip}
    \includegraphics[width=1\columnwidth,keepaspectratio]{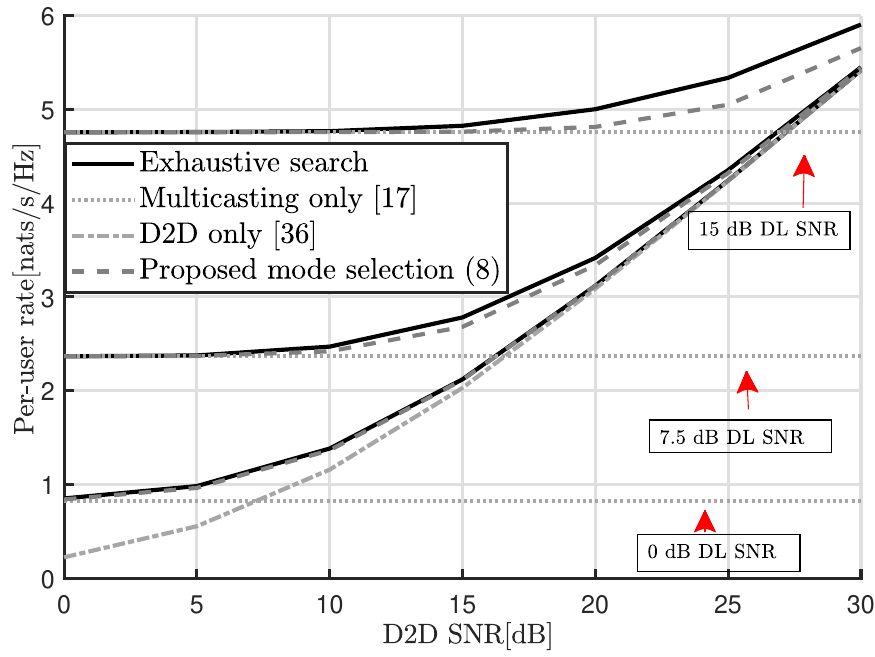}
    \caption{Average delivery rate Vs. D2D and DL SNR for $K=4$, $L=2$  $\tau=2$, and $m=0$.}
\label{fig:DIFD2DSINR}
\end{figure}
%Add here justification for the second channel model... in order to , we now consider a simplified blaa vlaa
%we consider D2D and DL channels, where ${\mathbf h}_k \sim \mathbb{CN}(0,{\mathbf{I}}), \ \text{and} \ h_{ik} \sim \mathbb{CN}(0,1)$.

Fig.~\ref{fig:DIFD2DSINR} depicts the average delivery rate for different D2D and DL SNRs, where all the users experience similar link conditions. As illustrated, the gap between the exhaustive search and the proposed scheme remains negligible in different SNR regions. When the received SNR of DL and D2D links are within the same range, i.e., the rate difference is not significant, choosing the right beneficial D2D groups is challenging. However, the proposed scheme still follows the exhaustive search in such scenarios. Since the curves have similar behavior  for the two cases, we only represent the $K=4$ case.

Fig.~\ref{fig:DIFD2DSINR-GGS-10db} compares the proposed scheme for D2D group size $|\mathcal{D}|=\tau+1 = 3$ (proposed mode selection~\eqref{eq:therishold}) versus general D2D group size (proposed mode selection~\eqref{eq:general-c-theri}) for different D2D received SNRs.\footnote{Due to an excessively large number of different subset combinations, performing the exhaustive search is computationally challenging ($\sim2^{10}$ different cases must be evaluated); thus, only the two thresholds are being compared.} When users experience similar channel conditions, the authors in~\cite{Ji2016} show that the achievable rate of their proposed D2D scheme, which is also performed in the D2D mode in this paper, is within a constant factor from the optimal value. The numerical example shown in Fig.~\ref{fig:DIFD2DSINR-GGS-10db} further demonstrates that when users experience similar channel conditions, considering $|\mathcal{D}|<\tau+1$ does not significantly improve the total delivery time (which corresponds to the results in~\cite{Ji2016}). 

\begin{figure}[t!]
\begin{minipage}[c]{0.49\textwidth}
\centering
\setlength\abovecaptionskip{-0.25\baselineskip}
\includegraphics[width=\linewidth]{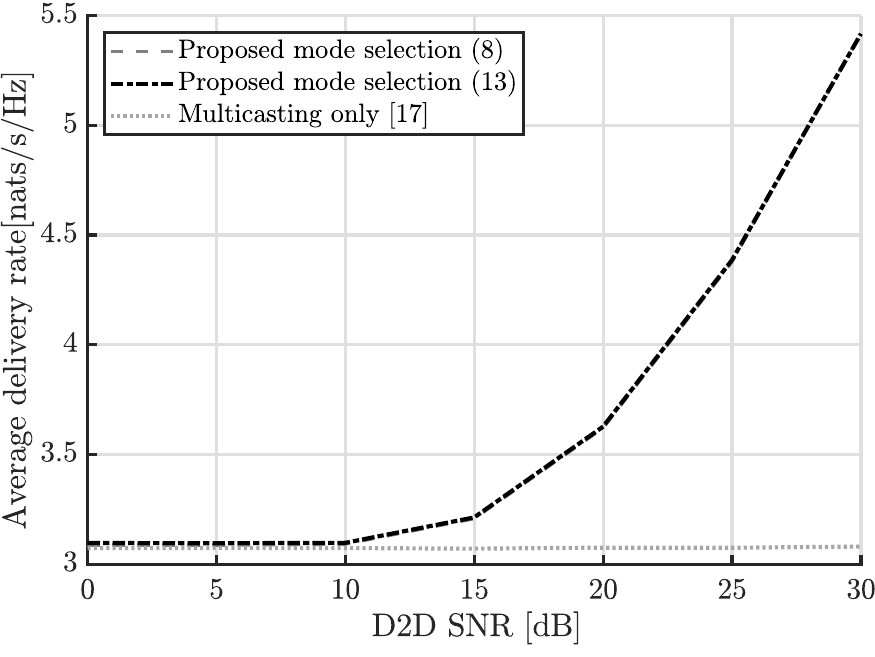}
\caption{Average delivery rate Vs. D2D SNR for $K=4$, $\tau=2$, $L=2$, and $10$ dB DL SNR.}
\label{fig:DIFD2DSINR-GGS-10db}
\end{minipage}
\hspace{1mm}
\begin{minipage}[c]{0.49\textwidth}
\centering
\setlength\abovecaptionskip{-0.25\baselineskip}
\includegraphics[width=\linewidth]{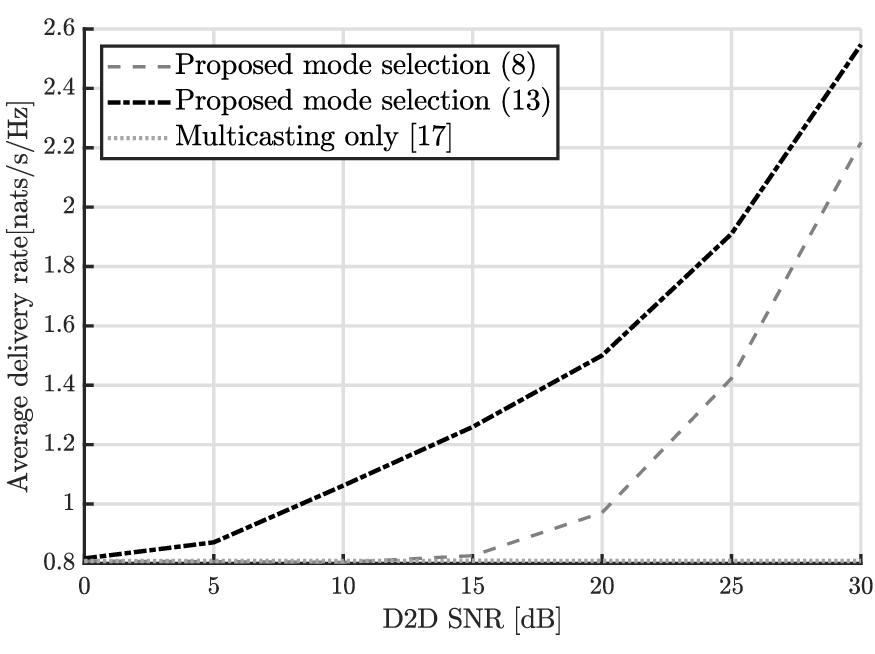}
\caption{Average delivery rate Vs. D2D SNR for $K=4$, $\tau=2$, $L=2$, and $0$ dB DL SNR. User pairs $(1,3)$ and $(2,4)$ are attenuated by $10$ dB.}
\label{fig:DIFD2DSINR-GGS-0db}
\end{minipage} 
\end{figure}

Note that the results in~\cite{Ji2016} are valid for the error-free D2D links with constant link capacity; however, when users experience different D2D link capacity, the results in~\cite{Ji2016}
do not hold anymore. For instance, the message $\Tilde{X}_{1,2,3}$ in section~\ref{sec:Example2} can still be transmitted through two D2D groups, e.g., $(1,2)$ and $(1,3)$, even if user pair $(3,2)$ is not in a favorable D2D condition (i.e., D2D group \{1,2,3\} does not satisfy Eq.~\eqref{eq:therishold}). Fig.~\ref{fig:DIFD2DSINR-GGS-0db} illustrates the average delivery rate for the same scenario as in Fig.~\ref{fig:DIFD2DSINR-GGS-10db}, but when D2D user pairs $(1,3)$ and $(2,4)$ are attenuated by 10 dB, i.e., %$g_{ik} \sim \mathbb{CN}(0,1)$ for 
user pairs $(1,2), \ (1,4), \ (2,3)$, and $ (3,4)$ experience similar D2D SNRs while %$g_{ik} \sim \mathbb{CN}(0,0.1)$ for 
D2D SNR for user pairs $(1,3)$ and $(2,4)$ is $10$ dB less than the rest of the D2D pairs. The results show that {considering D2D groups with various sizes is crucial when} users experience uneven D2D link conditions.

%%%%%%%%%%%%%%%%%%%%%%%%%%%%%%%%%%%%%%%%%%%%%%%%%%%%%%%%%%%%%%%%%%%%%%%%%%%%%%%%%
\section{Conclusion}
\label{sec:conclusions}
%%%%%%%%%%%%%%%%%%%%%%%%%%%%%%%%%%%%%%%%%%%%%%%%%%%%%%%%%%%%%%%%%%%%%%%%%%%%%%%%%
A novel cache-aided delivery method, {comprising orthogonal D2D and DL transmission phases}, was proposed to improve the {multiantenna CC- based content delivery}. In the proposed method, the DL multicasting of file fragments is complemented by allowing direct {D2D} exchange of local cache content. The benefits of partial D2D offloading the content were investigated. We showed that the benefits of using D2D transmission are two-fold. First, nearby users may greatly benefit from direct content exchange providing a notably increased overall delivery rate. Second, the partial offloading of the contents in the D2D phase decreases the DL beamforming complexity. This is due to the reduced number of variables and conditions in the beamformer optimization problem.

We showed that the D2D optimal subset selection imposes high computational complexity due to the DL multicast design. On the other hand, it is a combinatorial problem and hence NP-hard. Therefore, we proposed an approximation for the DL achievable rate without computing beamformers to overcome these practical limitations. Next, based on the approximated DL rate, we provided a low complexity mode selection algorithm, which allows efficient determination of D2D opportunities even for many users. The simulation results demonstrated that the proposed heuristic mode selection performs comparably to the exhaustive search with significantly reduced complexity. Furthermore, further extension is possible by considering the users' energy expenditure in the D2D mode and the energy efficiency of the BS in the DL mode. {Another direction includes applying ML-based tools for cache placement and D2D/DL mode selection. Particularly, users' movement patterns can be predicted using past data history, which can then be used to decide on partial file delivery in D2D mode.}
{Finally}, further overall improvement can be achieved by allowing parallel mutually interfering transmissions within multiple D2D groups.

%%%%%%%%%%%%%%%%%%%%%%%%%%%%%%%%%%%%%%%%%%%%%%%%%%%%%%%%%%%%%%%%%%%
\appendices \section{Proof of Theorem 1}
\label{sec:app}

First, consider the case where no D2D transmission is available. In this case, the total number of messages transmitted to all the users is ${M}_{\text{T}}=\binom{\tau+L}{\tau+1}$, and each user needs $W=\binom{\tau+L-1}{\tau}$ of these messages to decode its intended file. Then, the per-user number of MAC conditions is $J(\tau,i=0,m=0,L) = 2^{W}-1$. Accordingly, considering the $\tau+L$ number of users served in each transmission, the total number of MAC conditions in the beamformer optimization problem is $M(\tau,i=0,m=0,L) = (\tau+L)J_0 $.

Now, when $i$ user groups of size $\tau+1$ are selected for the D2D phase, the total number of MAC conditions varies based on which user groups are selected. To illustrate, consider the simple scenario where $i=2$ and $\{\mathcal{S}_1 \subset [K] : \ |\mathcal{S}_1|=\tau+1\}$ and $\{\mathcal{S}_2 \subset [K] : \ |\mathcal{S}_2|=\tau+1\}$ denote the first and the second D2D groups, respectively. Then, the total number of MAC conditions varies based on the number of users in common between the two groups. In this regard, denoting $c_{u}$ as the number of users in common, i.e., $c_{u} = |\mathcal{S}_1 \cap \mathcal{S}_2|$, the total number of MAC conditions varies as follows.

\subsubsection*{Case-1 ($c_{u} = 0$)}
In this case, all the users in the two D2D groups receive one subfile. Thus, they need $W-1$ messages in the DL signal to decode their intended files. Accordingly, the number of MAC conditions for these users is $J(\tau,i=1,m=0,L) = 2^{W-1}-1$, and the total number of MAC conditions (in this case) is 
%\begin{equation}\nonumber
    $M(\tau,i=1,m=0,L)  = (L-\tau-2)J(\tau,i=0,m=0,L) +2(\tau+1)J(\tau,i=1,m=0,L) \approx M(\tau,i=0,m=0,L)-2(\tau+1)J(\tau,i=1,m=0,L)$.
%\end{equation}
\subsubsection*{Case-2 ($c_{u} \neq 0$)}
Denote $\mathcal{S}_c = \mathcal{S}_1 \cap \mathcal{S}_2$ and $\mathcal{S}_r = \mathcal{S}_1 \cup \mathcal{S}_2 \setminus \mathcal{S}_1 \cap \mathcal{S}_2$ as the set of common/joint and uncommon/disjoint users, respectively. Then, in this case, the common users receive two subfiles, and the rest of the users in set $\mathcal{S}_r$ receive one subfile through the D2D phase. Thus, the common users need $W-2$ number of messages in the DL phase, and the number of MAC conditions for these users is $J(\tau,i=2,m=0, L) = 2^{W-2}-1$, while the rest require $W-1$ number of subfiles in the DL phase. Accordingly, the total number of MAC conditions (in this case) is
%\begin{equation}\nonumber
    $M(\tau,i=2,m=0, L) = (L-\tau-2+ c_{u}))J(\tau,i=0,m=0, L) +(2\tau+2 - c_{u})J(\tau,i=1,m=0, L) + c_{u}J(\tau,i=2,m=0, L) \approx M(\tau,i=1,m=0, L)+c_{u}J(\tau,i=2,m=0, L)$.
%\end{equation}

Compared to the first case, the number of MAC conditions is more due to having common users. Therefore, when D2D transmissions occur uniformly among all users, the total number of MAC conditions is minimized. In other words, when all the users need almost the same number of subfiles in the DL  phase, the number of MAC conditions is minimized. On the other hand, when D2D transmissions occur to a limited number of users, the number of MAC conditions is maximized. 

\subsection{Minimum number of MAC conditions}
Let us define a super set $\Bar{\mathcal{V}}$ which includes all the D2D sets with size less than $\tau+1$, i.e., each element in $\Bar{\mathcal{V}}$ is a set $\mathcal{V} \subset [K]$ such that $2 \leq |\mathcal{V}| < \tau+1$. We also define another super set $\Bar{\mathcal{D}}$, where each element in $\Bar{\mathcal{D}}$ is a set $\mathcal{D} \subseteq [K]$ such that $|\mathcal{D}| = \tau+1$, i.e.,
\begin{equation} \label{eq:full-size-d2d-group}
    \Bar{\mathcal{D}} = \{\mathcal{D} \,|\,\mathcal{D} \subseteq [K], \,|\mathcal{D}| = \tau+1\}.
\end{equation}
Now, assume $m$ subfiles are delivered through D2D groups in $\Bar{\mathcal{V}}$, and $(\tau+1)i$ subfiles are delivered through D2D groups in $\Bar{\mathcal{D}}$, where $|\Bar{\mathcal{D}}| = i$. In this case, the total number of remaining subfiles in the DL signal is $(\tau+1) \Big(\binom{\tau+L}{\tau+1}-i\Big)-m$.
Based on the previous paragraph, when these subfiles are uniformly allocated to all users, the total number of MAC conditions in the beamformer design is minimized. Therefore, in such cases, %when the fraction $a=\frac{(\tau+1)\Big(\binom{\tau+L}{\tau+1}-i\Big)-m}{\tau+L}$ is an integer, all the users receive $a$ subfiles in the DL phase to decode their intended files. However, when $a$ is not an integer, 
$\tau+L-b$ users receive $a=\left\lfloor{\frac{(\tau+1)(\binom{\tau+L}{\tau+1}-i)-m}{\tau+L}}\right\rfloor$ subfiles and $b\triangleq\left(\tau+1\right)\left({M}_{\text{T}}-i\right)-m-a\left(\tau+L\right)$ users receive $a+1$ subfiles in the DL phase. Consequently, the minimum number of MAC conditions in the DL phase is
%\begin{equation} \nonumber
    $\underline{M}(\tau,i,m, L)=(\tau+L-b)(2^a-1)+b(2^{a+1}-1)$.
%\end{equation}

\subsection{Maximum number of MAC conditions}

Let us first consider the $m=0$ case; as discussed earlier, the number of MAC conditions is maximized when D2D groups are selected from a limited number of users. Thus, to maximize the number of MAC conditions, we first denote $\mathcal{U} \subseteq [K]$ as the smallest set of users to form $i$ different D2D groups with size $\tau+1$, i.e., $U = \arg\min \binom{\hat{U}}{\tau +1}$ such that $ \binom{\hat{U}}{\tau +1} \geq i$, where $U = |\mathcal{U}|$. Clearly, $U$ is not smaller than $\tau+1$ based on the definition. Then, we define $\Bar{\mathcal{C}}_n(\mathcal{A}) = \{\mathcal{B} \ | \ \mathcal{B} \subseteq \mathcal{A}, \ |\mathcal{B}|=n\}$ as the collection of all the subsets of size $n$ from the set $\mathcal{A}$, where $|\Bar{\mathcal{C}}_n(\mathcal{A})| = \binom{|\mathcal{A}|}{n}$. We also denote all the non-empty subsets of the set $\mathcal{U}$ as $\mathcal{S}$ i.e., $\mathcal{S} \subseteq \mathcal{U}, |\mathcal{S}| \geq 1$. Moreover, we call such a set $\mathcal{S}$ to be a \textit{utilized-D2D-set} if $\Bar{\mathcal{C}}_{\tau+1}(\mathcal{S}) \subset \Bar{\mathcal{D}}$, where $\Bar{\mathcal{D}}$ is defined in~\eqref{eq:full-size-d2d-group}. We represent the set of all the \textit{utilized-D2D-set}s as $\Bar{\mathcal{S}}(\mathcal{U})$, i.e.,
$\Bar{\mathcal{S}}(\mathcal{U})=\{\mathcal{S} \ | \ \mathcal{S} \subseteq \mathcal{U}, \ \Bar{\mathcal{C}}_{\tau+1}(\mathcal{S}) \subset \Bar{\mathcal{D}}\}$. We are interested in the \textit{utilized-D2D-set} with the largest number of users. We call such a set as the \textit{fully-utilized-D2D-set} and denote it with $\hat{S}$. In other words, $\hat{S} \in \Bar{\mathcal{S}}(\mathcal{U})$ and for every $\mathcal{S}_0 \in \Bar{\mathcal{S}}(\mathcal{U})$ we have $|\hat{S}| \geq |\mathcal{S}_0|$. Therefore, by the definition, $|\hat{\mathcal{S}}| = U-1$, and the users belonging to the set $\hat{\mathcal{S}}$ form $\binom{U-1}{\tau +1}$ D2D groups. The remaining $X=i  -\binom{U-1}{\tau+1}$ D2D groups include $\tau$ number of users from the set $\hat{\mathcal{S}}$ and one remaining user denoted by $k_r$, where $k_r = \mathcal{U} \setminus \hat{\mathcal{S}}$.

Therefore, in this case, there exist $3$ different types of users (see Fig.~\ref{fig: mac-upper-bound})
\begin{itemize}
\item \textbf{First type}, Those who have not received any subfile in D2D transmissions. So they need $W = \binom{\tau+L-1}{\tau}$ number of subfiles in DL phase, and their total number is $K_0 = \tau+L-U$.
\item \textbf{Second type}, Those who have received all the available subfiles in the set $\hat{\mathcal{S}}$. %but did not receive any subfiles in subsets consisting of the new user. 
Thus, they need $W_{\hat{\mathcal{S}}} =  \binom{\tau+L-1}{\tau}-\binom{(U-1)-1}{\tau}$ number of subfiles in the DL phase, and their total number is $K_1 = U-1$. For convenience, we ignore the received subfiles in the $X$ remaining D2D groups for $\forall k \in \hat{\mathcal{S}}$.
\item \textbf{Third type}, The user $k_r$ who has received $X$ number of subfiles in the groups consisting of itself and $\tau$ members of the set $\hat{\mathcal{S}}$, so it needs $W_{k_r} = \binom{\tau+L-1}{\tau}-X$ number of subfiles in the DL phase.
\end{itemize}
Therfore, the maximum number of MAC conditions for $|\Bar{\mathcal{D}}|=i$ and $m=0$ is upper bounded by
%\begin{align} 
        %\text{MAC}^{i}_\text{max}= \ &(\tau+L-U)(2^{\binom{\tau+L-1}{\tau}}-1) + (U-1)(2^{(\binom{\tau+L-1}{\tau}-\binom{U-2}{\tau})}-1)+ (2^{\binom{\tau+L-1}{\tau}-X}-1)\text{.}\nonumber \nonumber 
        $\overline{M}(\tau,i,m, L)= K_{0}(2^{W}-1) + K_{1}(2^{W_{\hat{\mathcal{S}}}}-1)+ (2^{W_{k_r}}-1)\text{.}$ %\nonumber \nonumber
    %\\ &+U_{1}(2^{(\binom{\tau+L-1}{\tau}-(\binom{U-2}{\tau}+\binom{U_{1}-1}{\tau-1}-Y))}-1)+ (2^{\binom{\tau+L-1}{\tau}-X}-1)\text{.}\nonumber
%\end{align}

For the case $m \neq 0$, suppose each D2D subset $\mathcal{V}$ delivers a single subfile to each of the users $k \in \mathcal{V}$; accordingly, each user $k \in \mathcal{V}$ requires one less subfile in the DL phase. Based on the previous discussion, removing $v$ subfiles from one user reduces the total number of MAC conditions less than removing one subfile for $v$ users, i.e.,
\begin{equation*}
    M(\tau,i,0,L)-v2^{x-1} < M(\tau,i,0,L)-2^{x-v}.
\end{equation*} %Note that the slope of the function $2^{x-i}$ is monotonically decreasing with respect to $i$. Thus, choosing a limited number of users for {D2D} transmission results in a smaller reduction in the total number of {MAC} conditions. 
Therefore, we assume {that} each D2D subset of size $|\mathcal{V}| = v$ delivers all the subfiles to a single user. Also, to consider as limited D2D user sets as possible, we assume these $m$ subfiles are also shared among the second type of users. In this regard, considering total number of $m$ subfiles delivered through D2D groups in $\Bar{\mathcal{V}}$, the minimum number of served users is $\phi = \left\lfloor \frac{m}{W_{\hat{\mathcal{S}}}} \right\rfloor$. Therefore, the maximum number of MAC conditions for $|\Bar{\mathcal{D}}|=i$ and $m \neq 0$ is upper bounded by 
%\begin{align} 
        $\overline{M}(\tau,i,m, L)= \ K_{0}(2^{W}-1) + (K_{1}-\phi)(2^{W_{\hat{\mathcal{S}}}}-1)+ (2^{W_{k_r}}-1)\text{.}$%\nonumber \nonumber
%\end{align}
\begin{figure}
    \centering 
    \setlength\abovecaptionskip{-0.25\baselineskip}
    \includegraphics[width=1\columnwidth,keepaspectratio]{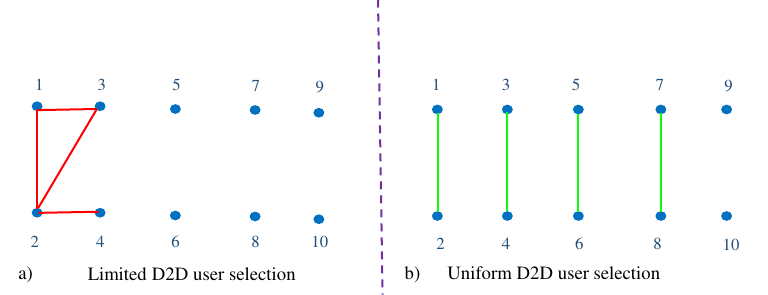}
    \caption{a) The limited D2D user selection case, where $K=10$, $ \tau=1$, $i=4$, $\Bar{\mathcal{D}} = \{\{1,2\}, \{1,3\}, \{2,3\}, \{2,4\}\},$ $ \mathcal{U}=\{1,2,3,4\},$ $\hat{\mathcal{S}}=\{1, 2, 3\}$, and $k_r = 4$. b) Uniform D2D subset selection case.}
\label{fig: mac-upper-bound}
\end{figure}
%%%%%%%%%%%%%%%%%%%%%%%%%%%%%%%%%%%%%%%%%%%%%%%%%%%%%%%%%%%%%%%%%%%%%%%%%%%%%%%%
\section{Proof of Theorem 2}
\label{sec:appB}
%%%%%%%%%%%%%%%%%%%%%%%%%%%%%%%%%%%%%%%%%%%%%%%%%%%%%%%%%%%%%%%%%%%%%%%%%%%%%%%%
When D2D transmission is not feasible, the BS {transmits} ${M}_{\text{T}}$ messages in the DL phase, where each user needs $W$ number of these messages and considers the rest (i.e., ${M}_{\text{T}}-W$ terms) as interference. In this case, for each intended message $\mathcal{D}$  of user $k$ ($\mathcal{D}$ is the message index), one quadratic term for the intended message ($|{\mathbf{h}}_{k}^{H}{\mathbf{ w}}_{\mathcal{D}}|^2, \ \mathcal{D} \in \Omega_{k}^{\mathcal{S}}$) plus ${M}_{\text{T}}-W$ quadratic terms for the interfering messages ($|{\mathbf{h}}_{k}^{H}{\mathbf{w}}_{\mathcal{D'}}|^2$, \ $\mathcal{D'} \in \mathcal{I}_k$) {is} considered in~\eqref{Eq:MAC_general}. Thus, in total, the BS {considers} $Q_{k}=W(M_{\text{T}}-W+1)$ quadratic terms in its' optimization problem for each user. Therefore,
the total number of quadratic variables $Q_{k}$ is a quadratic function in terms of $W$ (see Fig. \ref{fig:quadbehav} (a)), which is maximized when $\Bar{W} = \frac{M_{\text{T}}+1}{2}$. Moreover, when D2D transmission is not feasible, $W$ is equal to $\binom{\tau+L-1}{\tau}=\frac{\tau+1}{\tau+L}M_{\text{T}}$, where $W$ is greater than $\Bar{W}$ for $1 \leq L < \tau+2$, and $W$ is less than $\Bar{W}$ for $\tau+2 < L$ (see Fig.~\ref{fig:quadbehav} (a)).

For the case $|\Bar{\mathcal{D}}| = i$ (assume $m=0$), let us define $W^{i}_{k}$ as the number of DL messages needed by user $k$ after $i$ D2D time slots. We also define  $M^{i}_{\text{T}}$ as the total number of transmitted messages in the DL phase after $i$ D2D time slots. In this case, $W^{i}_{k}$ is independent of $M^{i}_{\text{T}}$, i.e., $0 \leq \frac{W^{i}_{k}}{M^{i}_{\text{T}}} \leq 1$, and $Q^{i}_{k} = W^{i}_{k}(M^{i}_{T}-W^{i}_{k}+1)$ is a quadratic function in terms of $W^{i}_{k}$.

We consider two extreme scenarios in this case (similar to Appendix~\ref{sec:app}): 
\begin{enumerate}
\item When the D2D subsets are selected uniformly among all the users.

\item When the D2D subsets are selected among a limited number of users.
\end{enumerate}
In the first scenario, since all the users have received almost the same number of subfiles in the D2D phase, they all need almost the same number of messages in the DL phase (see Fig.~\ref{fig:quadbehav} (b)). However, in the second scenario, since some users have received most of their intended files in the D2D phase, they need a few numbers of messages in the DL phase. In contrast, users who did not receive any subfile in the D2D phase need most of the subfiles in the DL phase. Thus, in this scenario, $W^{i}_{k}$s are either on the right-hand side of the maximum point or on the left-hand side of it (see Fig.~\ref{fig:quadbehav} (c)). %Fig. \ref{fig:quadbehav} (b) and (c) show the two extreme scenarios for some particular $\tau$ and $L$.
\begin{figure*}[ht]
    \centering 
    \setlength\abovecaptionskip{-0.25\baselineskip}
    \includegraphics[width=2\columnwidth,keepaspectratio]{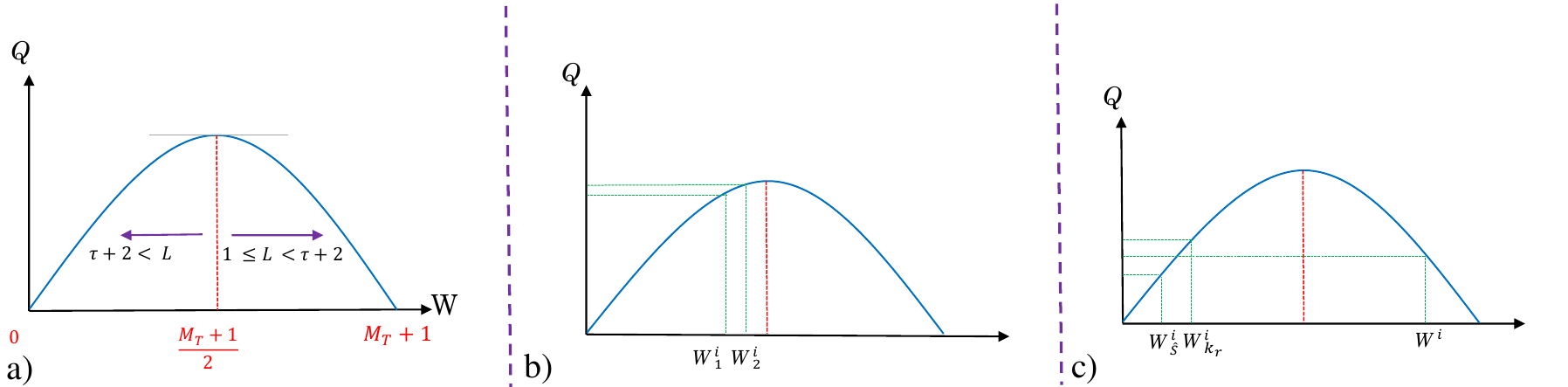}
    \caption{Total number of quadratic variables for each user in different scenarios: a) no D2D case b) uniform scenario c) scenario with {a} limited number of users.}
\label{fig:quadbehav}
\end{figure*}

Therefore, these two cases are two extreme cases for the total number of quadratic variables as well. In this regard, {the total number of quadratic variables is maximized when D2D subsets are uniformly selected among all the users.} On the other hand, when a limited number of users are selected for D2D transmissions, $Q$ is minimized. Finally, the total number of quadratic variables (after $i$ D2D time slots) is computed as 
%\begin{equation} \nonumber    
$Q^{i}=\sum_{k \in [\tau+L]}Q^{i}_{k}$.
%\end{equation}
Substituting the total number of needed messages for each user, defined in appendix~\ref{sec:app}, equations \eqref{eq:NA9} and \eqref{eq:NA11} are~achieved.

For the case $m\neq 0$, $M^{i}_{T}$ can be lower-approximated by $\underline{M_\text{T}}= \left\lceil\frac{\big(\tau+1\big)\big(\binom{\tau+L}{\tau+1}-i\big)-m}{\tau+1}\right\rceil$, where we assume each $\tau+1$ number of transmitted subfiles through D2D groups in $\Bar{\mathcal{V}}$ removes one of the remaining messages in the DL phase. The $M^{i}_{T}$ is upper-approximated by $\overline{M_\text{T}} = \binom{\tau+L}{\tau+1}-i$, where we assume no messages are being eliminated after D2D transmissions in D2D groups with size less than $\tau+1$.
In this case, the total number of quadratic variables is computed as same as the $m=0$ case by substituting $M^{i}_{T}$ with $\overline{M_\text{T}}$ ($\underline{M_\text{T}}$) for upper bound (lower bound).

\bibliographystyle{IEEEtran}
\bibliography{IEEEabrv,conf_short,jour_short,references}
\end{document}